\documentclass[letterpaper,11pt]{article}
\usepackage{booktabs} 
\usepackage{booktabs} 
\usepackage[ruled]{algorithm2e} 

\usepackage[margin=1in]{geometry}

\usepackage[english]{babel}

\usepackage{bm}
\usepackage{amsfonts,amsmath,amssymb,amsthm}

\usepackage{enumitem}

\usepackage[numbers]{natbib}
\usepackage{comment}
\usepackage{graphicx}
\usepackage{subcaption}

\usepackage{algorithmic}
\usepackage{tablefootnote}

\newtheorem{theorem}{Theorem}[section]
\newtheorem{lemma}[theorem]{Lemma}

\newcommand{\opt}[1]{\textsc{Opt$\left(#1\right)$}}
\newcommand{\rev}[1]{\textsc{Rev$\left(#1\right)$}}
\newcommand{\dskl}[1]{D_{SKL}\left(#1\right)}

\newcommand{\treps}[1]{t^{\min}_{\epsilon}(#1)}
\newcommand{\trv}[1]{t^{\max}_{\bar{\bm{v}}}(#1)}
\newcommand{\trvi}[1]{t^{\max}_{\bar{v}_i}(#1)}

\newcommand{\E}[1]{\textsc{\textbf{E}$\left[#1\right]$}}

\newcommand{\df}{\text{d$\,$}}

\newcommand{\bn}{n}

\newcommand{\vvi}{\varphi_i}

\newcommand{\ep}{\epsilon}

\newcommand{\shaded}{d_{f}}
\newcommand{\shades}{s_{f}}

\newcommand{\defeq}{:=}
\renewcommand{\vec}[1]{\bm{#1}}

\newcommand{\tcell}[1]{\begin{tabular}[c]{@{}c@{}} #1 \end{tabular}}

\SetAlFnt{\small}
\SetAlCapFnt{\small}
\SetAlCapNameFnt{\small}
\SetAlCapHSkip{0pt}
\IncMargin{-\parindent}

\title{Targeting Makes Sample Efficiency in Auction Design
\thanks{To appear in The Twenty-Second ACM Conference on Economics and Computation (EC'21), 2021}
}

\author{
    Yihang Hu
    \thanks{IIIS, Tsinghua University. Email: huyx17@mails.tsinghua.edu.cn}
    \and
    Zhiyi Huang
    \thanks{The University of Hong Kong. Email: zhiyi@cs.hku.hk. This work is supported in part by Research Grants Council of Hong Kong, Grant HKU17203717E.}
    \and
    Yiheng Shen
    \thanks{IIIS, Tsinghua University. Email: shen-yh17@mails.tsinghua.edu.cn}
    \and
    Xiangning Wang
    \thanks{The University of Hong Kong. Email: xnwang@cs.hku.hk}
}
\begin{document}

\begin{titlepage}
    \thispagestyle{empty}

    \maketitle 

    \begin{abstract}
        This paper introduces the targeted sampling model in optimal auction design.
In this model, the seller may specify a quantile interval and sample from a buyer's prior restricted to the interval.
This can be interpreted as allowing the seller to, for example, examine the top $40\%$ bids from previous buyers with the same characteristics.
The targeting power is quantified with a parameter $\Delta \in [0, 1]$ which lower bounds how small the quantile intervals could be.
When $\Delta = 1$, it degenerates to Cole and Roughgarden's model of i.i.d.\ samples;  
when it is the idealized case of $\Delta = 0$, it degenerates to the model studied by \citet{chen2018brief}. 
For instance, for $n$ buyers with bounded values in $[0, 1]$, $\tilde{O}(\epsilon^{-1})$ targeted samples suffice while it is known that at least $\tilde{\Omega}(n \epsilon^{-2})$ i.i.d.\ samples are needed.
In other words, targeted sampling with sufficient targeting power allows us to remove the linear dependence in $n$, and to improve the quadratic dependence in $\epsilon^{-1}$ to linear.
In this work, we introduce new technical ingredients and show that the number of targeted samples sufficient for learning an $\epsilon$-optimal auction is substantially smaller than the sample complexity of i.i.d.\ samples for the full spectrum of $\Delta \in [0, 1)$.
Even with only mild targeting power, i.e., whenever $\Delta = o(1)$, our targeted sample complexity upper bounds are strictly smaller than the optimal sample complexity of i.i.d.\ samples.

\end{abstract}
\end{titlepage}

\section{Introduction}
\label{sec:intro}

Suppose that Alice is going to auction a bottle of collectible wine to Bob, software engineer from San Francisco, Charlie, businessman from New York, and Dave, real-estate broker from Chicago.
How should Alice design her auction to maximize revenue?

This is a special case of multi-buyer single-item auctions, when the number of buyers is $n = 3$.
\citet{Mye81} characterized the optimal auction decades ago under the assumption that the buyers' values for the item are drawn independently (but not necessarily identically) from some prior $\vec{D} = D_1 \times D_2 \times \dots \times D_n$.
However, the solution by Myerson requires fine-grained knowledge about the distributions, including their cumulative distribution function (cdf) and probability density function (pdf).
Such information is rarely available in practice.



Although precise knowledge of the prior is unavailable, one may hope that an approximate estimation would suffice for designing near optimal auctions.
For instance, Alice may have access to past bids by software engineers from San Francisco, businessmen from New York, and real-estate brokers from Chicago for the same wine and the same year, so that she could estimate the value distributions of Bob, Charlie, and Dave.
Indeed, \citet{OstrovskyS-EC-2011} successfully applied auction theory to internet advertising auctions using an approach in this spirit.
Moreover, \citet{CR14} viewed the data as independent and identically distributed (i.i.d.) samples and investigated the sample complexity:
\emph{How many i.i.d.\ samples per buyer are sufficient and necessary for learning an auction that is optimal up to $\epsilon$?}
Culminating in a series of studies, Guo et al.~\cite{GHZ19} recently pinned down the optimal sample complexity up to logarithmic factors.
For instance, if the values are bounded between $0$ and $1$, the optimal sample complexity is $\tilde{\Theta} \big( n \epsilon^{-2} \big)$.

We see two drawbacks in the above model of i.i.d.\ samples.
On the one hand, there may not be any single comprehensive dataset that could provide data equivalent to i.i.d.\ samples.
Instead, one may need to aggregate information from multiple datasets, one for senior software engineers, one for junior software engineers, etc.

On the other hand, the lower bound instance by \citet{GHZ19} suggests that i.i.d.\ samples are wasteful in the sense that most samples are irrelevant for learning a near optimal auction.
In a nutshell, the lower bound considers $n$ buyers with two types of distributions such that:
(1) a buyer gets the item in the optimal auction only if its value is in the top $\frac{1}{n}$ portion, and if its prior is of the first type;
(2) the two types of distributions were identical except in the top $\frac{1}{n}$ portion; and
(3) even conditioned on being in the top $\frac{1}{n}$, the distributions are similar, requiring at least $\epsilon^{-2}$ samples to distinguish apart.
The first condition forces the seller to correctly distinguish the two types of buyers.
The second means that only $1$ in $n$ samples are useful; the third states that we need $\epsilon^{-2}$ useful samples.

Therefore, this paper considers aggregating samples from conditional distributions of the priors, which correspond to the samples from different datasets, or targeted sampling queries supported by some dataset.
In particular, we focus on a benign special case in which the seller may sample conditioned on any quantile interval.
\emph{Can we use fewer samples by targeting the top part of the distributions?}
In other words, would it be more efficient for Alice to look at a data set of past bids by, say, top $40\%$ of software engineers from San Francisco, etc.?%

Although the data markets may not widely support the above targeted sampling queries at the moment, our results demonstrate their improved sample efficiency and therefore, suggest that it may be beneficial to support them in the future.
Further, targeted or conditional samples have been studied in other problems in theoretical computer science, notably in property testing (e.g., \citet{canonne2015testing, chakraborty2016power}).
Finally, we leave as a future direction how to aggregate samples from an arbitrary collection of conditional distributions of the priors.

\subsection{Our Contributions}

\subsubsection*{Targeted Sampling Model}
We initiate the study of targeted sampling in optimal auction theory.
Concretely, we propose a model where the seller may access the prior of each buyer via a targeted sampling oracle, which takes a quantile interval as input and returns a sample value within the interval.
For example consider the lower bound instance by Guo et al.~\cite{GHZ19}.
Targeted sampling allows us to directly collect samples conditioned on being in the top $\frac{1}{n}$ portion and thus, reducing the number of samples by a factor $n$.

We further quantify the targeting power with a parameter $0 < \Delta < 1$ such that the targeting quantile intervals must have length at least $\Delta$.
The exclusion of arbitrarily small intervals is motivated by potential privacy regulation, and the observation that the data holders in information markets may not be able to answer arbitrarily fine-grained queries in general.

We consider the \emph{targeted sample complexity} as a function of the number of buyers $n$, the approximation parameter $\epsilon$, the targeting power parameter $\Delta$, and other parameters of the distribution family under consideration: 
\emph{How many targeted samples per buyer are sufficient and necessary for learning an auction that is optimal up to $\epsilon$?}

In the extreme case when $\Delta = 0$, 
the seller may query a quantile interval that is a single point to ask about a buyer's value at a specific quantile $q$ in its prior.
We call it the \emph{targeted query} model and define \emph{targeted query complexity} accordingly. 
Targeted queries are stronger than targeted samples since the former can simulate the latter by first sampling a quantile $q$ from the given interval and then querying the value at $q$.
The main body of the paper will first focus on the idealized model of targeted queries to illustrate our techniques in Section~\ref{sec:multi}. 
It is a stepping stone towards the general model when the sampling intervals are not arbitrarily small, which is more realistic in our opinion. 
Then in Section~\ref{app:large-delta} we extends the techniques to obtain targeted sample complexity that degrades gracefully as the targeting power weakens.

On the other extreme when $\Delta = 1$, the model has no targeting power and degenerates to the Cole-Roughgarden model of i.i.d.\ samples.

\subsubsection*{Multiple Buyers}
We demonstrate the power of targeted sampling by showing that the targeted sample/query complexity is substantially smaller than the sample complexity in the Cole-Roughgarden model.
Again consider the case when values are bounded in $[0, 1]$ as a running example.
Our method shows that the targeted query complexity is at most $\tilde{O} \big( \epsilon^{-1} \big)$.
Comparing with the optimal sample complexity $\tilde{O}\big(n \epsilon^{-2}\big)$ in the Cole-Roughgarden model, the targeted query complexity removes the linear dependence in the number of buyers $n$ and improves the dependence in $\epsilon^{-1}$ from quadratic to linear.
This bound generalizes to targeted sampling for sufficiently small $\Delta = O \big( \frac{\epsilon}{n} \big)$;
for larger $\Delta$, the targeted sample complexity is at most $\tilde{O} \big( n\Delta \epsilon^{-2} \big)$.
This is smaller than the sample complexity in Cole-Roughgarden whenever $\Delta = o(1)$.
In other words, \emph{targeted sampling with even mild targeting power is strictly stronger than i.i.d.\ sampling.}
See Table~\ref{table:multi} for a comprehensive comparison of targeted versus i.i.d.\ samples for various families of distributions in the literature.

The targeted sample/query complexity bounds build on three technical ingredients.
Recall that in the lower bound example by Guo et al.~\cite{GHZ19} only the top $\frac{1}{n}$ portion of samples per buyer are useful.
However, this is not true in general. 
All samples of a buyer might be useful, for example, if it is the only buyer who could have a large value and the other buyers are irrelevant.
Nonetheless, i.i.d.\ samples are still wasteful in a different way:
the samples of only $1$ in $n$ buyers is useful.
Our first ingredient (Lemma~\ref{lem:theta}) proves that in general only $\tilde{O}\big(1\big)$ in $n$ i.i.d.\ samples matter \emph{on average}.

\paragraph{Very Few (Top) Samples Matter.}
There exist thresholds $\theta_i$, $1 \le i \le n$, such that:
\begin{enumerate}[label={(\arabic*)}, topsep=3pt, partopsep=0pt, itemsep=0pt, parsep=3pt]
    \item $\sum_{i=1}^n \theta_i = \tilde{O}(1)$, and
    \item almost all of the optimal revenue comes from the top $\theta_i$ portion of buyer $i$'s prior, $1 \le i \le n$.
\end{enumerate}

\medskip
\begin{table}[h]
\centering
\begin{tabular}{|c|c|c|c|}
\hline
\tcell{Family of\\ Distributions} &
\tcell{Sample Complexity\\ (Guo et al.~\cite{GHZ19})} &
\tcell{Targeted Query\\ Complexity\\ (\citet{chen2018brief})} &
\tcell{Targeted Sample\\ Complexity\tablefootnote{As discussed above, a sufficiently small $\Delta$ will lead to a degeneration to targeted query complexity. For simplicity we assume that $\Delta$ is not a sufficiently small one.}} \\
\hline
$[0, 1]$-bounded &
$\tilde{\Theta}(n\epsilon^{-2})$ &
$\tilde{O}(\epsilon^{-1})$ &
$\tilde{O}(\max\{ \epsilon^{-1}, n\Delta\epsilon^{-2}\})$ \\
\hline
$[1, H]$-bounded &
$\tilde{\Theta}(nH\epsilon^{-2})$ &
$\tilde{O}(\epsilon^{-1})$ &
$\tilde{O}(\max\{ H^{1/2}\epsilon^{-1}, n\Delta H\epsilon^{-2}\})$ \\
\hline
Regular &
$\tilde{\Theta}(n\epsilon^{-3})$ &
$\tilde{O}(\epsilon^{-1})$ &
$\tilde{O}(\max\{ \epsilon^{-3/2}, n\Delta\epsilon^{-3}\})$ \\
\hline
MHR &
$\tilde{\Theta}(n\epsilon^{-2})$ &
$\tilde{O}(\epsilon^{-1})$ &
$\tilde{O}(\max \{\epsilon^{-1}, n\Delta\epsilon^{-2}\})$ \\
\hline
\end{tabular}
\caption{Comparison of Cole-Roughgarden and targeted sampling models with $n$ buyers}
\label{table:multi}
\end{table}
As a thought experiment, suppose we further know the values of these thresholds.
Then, we may sample from the top $\theta_i$ portion of each buyer $i$'s prior, and estimate the priors by constructing the top $\theta_i$ portion from the targeted samples and letting the bottom $1 - \theta_i$ portion be a point mass at $0$.
The above result, however, is not constructive, so we do not know the thresholds.
This leads to two issues.
First, it is unclear which portion of the prior we shall sample from.
A natural solution is the doubling trick:
for each buyer $i$, guess $\theta_i = 1, \frac{1}{2}, \frac{1}{4}$, etc., and sample from the corresponding intervals;
then compute an estimation of $D_i$ aggregating all targeted samples.
The second issue is more subtle.
Even though the bottom $1 - \theta_i$ portion of buyer $i$'s prior contributes little to the revenue of the optimal auction, an inaccurate estimation of this portion of the prior may overestimate its contribution and incorrectly give the item to a buyer $i$ at a value in the bottom $1 - \theta_i$ portion.
This was not an issue in the thought experiment because the bottom $1 - \theta_i$ portion was rounded down to $0$.
The same treatment is infeasible when the thresholds are unknown.
To bound the revenue loss above, we generalize the analysis framework of Guo et al.~\cite{GHZ19}.
We write $D_i \succeq \tilde{D}_i$ (and $\tilde{D}_i \preceq D_i$) if $D_i$ first-order stochastically dominates $\tilde{D}_i$.

\paragraph{Sandwich Lemma.}
There are $\tilde{D}_i \preceq D_i$, $1 \le i \le n$ such that:
\begin{enumerate}[label={(\arabic*)}, topsep=3pt, partopsep=0pt, itemsep=0pt, parsep=3pt]
    \item the optimal revenue of $\vec{D}$ and $\vec{\tilde{D}}$ are $\epsilon$-close, and
    \item aggregating $\epsilon^{-2}$ samples from the top $\theta_i$ portion of buyer $i$'s prior for $\theta_i = 1, \frac{1}{2}, \frac{1}{4}$, etc., gives a dominated empirical distribution $E_i$ between $D_i$ and $\tilde{D}_i$, i.e., $D_i \succeq E_i \succeq \tilde{D}_i$.
\end{enumerate}

\medskip

By contrast, Guo et al.~\cite{GHZ19} used $\tilde{O} \big( n \epsilon^{-2} \big)$ i.i.d.\ samples to achieve a similar sandwich lemma.
The improvement comes from an application of the first ingredient, which allows us to satisfy the first condition with a $\vec{\tilde{D}}$ that deviates much further from $\vec{D}$ in the small-value regime compared to its counterpart in Guo et al.~\cite{GHZ19}.
Indeed, the choice of $\vec{\tilde{D}}$ is driven by the estimation error of a Bernstein's inequality and the doubling trick.
Then, the rest of the analysis follows by revenue monotonicity due to Devanur et al.~\cite{DHP16} like in the Cole-Roughgarden model:
running the optimal auction with respect to (w.r.t.) $\vec{E}$ gets at least the optimal revenue of $\vec{E}$ (strong revenue monotonicity), which is lower bounded by the optimal revenue of $\vec{\tilde{D}}$ (weak revenue monotonicity), which is $\epsilon$-close to the optimal revenue of $\vec{D}$ by the sandwich lemma.

The first two ingredients alone show a targeted sampling complexity upper bound of $\tilde{O}(\epsilon^{-2})$, removing the linear dependence in $n$.
Further, the algorithm is non-adaptive, and targets intervals only of the form $[0, q]$, the top $q$-portion, for $q = 1, \frac{1}{2}, \frac{1}{4}$, etc.

Finally, to further improve the dependence in $\epsilon^{-1}$ from quadratic to linear for targeted queries, and also for targeted samples with a sufficiently small $\Delta$, we adopt the $\vec{\tilde{D}}$ designed for the analysis of the doubling trick but abandon the doubling trick itself.
Instead, we directly design a targeted sampling algorithm that learns an $E_i$ sandwiched between $D_i$ and $\tilde{D}_i$, building on the following: 

\paragraph{Pinpoint Lemma.}
There are $N = \tilde{O}\big(\epsilon^{-1}\big)$ pinpoints $1 = q_0 > q_1 > \dots > q_N > q_{N+1}= 0$ such that for any buyer $i$ and any $j$, the value with quantile $q_j$ in $D_i$ has quantile at least $q_{j+1}$ in $\tilde{D}_i$.

\medskip

We demonstrate how to use the lemma with targeted queries for simplicity.
To get a distribution $E_i$ sandwiched between $D_i$ and $\tilde{D}_i$, it suffices to have $N = \tilde{O}\big(\epsilon^{-1}\big)$ queries at $q_1, q_2, \dots, q_N$ and construct a discrete distribution that agrees with $D_i$ at these pinpoints, plus a point mass at $0$.

\subsubsection*{Single Buyer}
In the special case of a single buyer, every auction can be interpreted as posting a take-it-or-leave-it price (potentially with randomness).
The expected revenue of posting a price $p$ is $p \cdot \Pr_{v \sim D} \big[ v \ge p \big]$.
It is common practice in the literature to view it as a function of the quantile $q(p) = \Pr_{v \sim D} \big[ v \ge p \big]$.
Then, the regular distributions are those whose revenue is a concave function of the quantile; 
further, distributions with monotone hazard rates (MHR) are those satisfying a certain sense of ``strict concavity''.
Exploiting the concavity of the revenue function, we further improve the query complexity for regular and MHR distributions to $\tilde{O} \big(1\big)$ in the single-buyer case.
The improved bounds also hold for targeted sample complexity for sufficiently small $\Delta$.
Finally, we complement the upper bounds with matching lower bounds.
See Table~\ref{table:single} for a comparison of results in the Cole-Roughgarden and targeted sampling models with a single buyer.

\begin{table}[t]
\centering
\begin{tabular}{|c|c|c|c|}
\hline
\tcell{Family of\\ Distributions} &
\tcell{Sample Complexity\\ (Huang et al.~\cite{HuangMR/2018/SICOMP})} & 
\tcell{Targeted Query\\ Complexity} &
\tcell{Targeted Sample Complexity\\ Upper bound\tablefootnote{Like in multiple buyers, we only list the results when $\Delta$ is not sufficiently small for the sake of simplicity.}} \\
\hline
$[0, 1]$-bounded &
$\Theta(\epsilon^{-2})$ &
$\tilde{\Theta}(\epsilon^{-1})$ &
$\tilde{O}(\max\{\epsilon^{-1},\Delta\epsilon^{-2})\})$\\
\hline
$[1, H]$-bounded &
$\tilde{\Theta}(H\epsilon^{-2})$ &
$\tilde{\Theta}(\epsilon^{-1})$ &
$\tilde{O}(\max\{H^{1/2}\epsilon^{-1},\Delta H\epsilon^{-2}\})$ \\
\hline
Regular &
$\tilde{\Theta}(\epsilon^{-3})$ &
$\tilde{\Theta}(1)$ & 
$\tilde {O}(\max \{1,\min\{\Delta^2\ep^{-4},\Delta\ep^{-3}\}\})$\\
\hline
MHR &
$\tilde{\Theta}(\epsilon^{-1.5})$ &
$\tilde{\Theta}(1)$ &
$\tilde{O}(\max \{1,\Delta^2\epsilon^{-2}\})$ \\
\hline
\end{tabular}
\caption{Comparison of Cole-Roughgarden and targeted sampling models with a single buyer}
\label{table:single}
\end{table}

\subsection{Related Works}

When Cole and Roughgarden~\cite{CR14} proposed their model of sample complexity of optimal auctions, they showed that for the regular and MHR distributions (more generally, $\alpha$-regular distributions) $\textrm{poly}(n,\epsilon^{-1})$ samples are sufficient and necessary for learning an $n$-buyer auction that is optimal up to a $(1 - \epsilon)$-approximation.
Subsequently, there has been a long line of research in this direction.
Notably, Morgenstern and Roughgarden~\cite{MR15}, Devanur et al.~\cite{DHP16}, Syrgkanis~\cite{syrgkanis2017sample}, and Gonczarowski and Nisan~\cite{GN17} improved the sample complexity upper bounds for regular and MHR distributions, and showed upper bounds for bounded-support distributions w.r.t.\ both multiplicative and additive approximations.
Recently, Guo et al.~\cite{GHZ19} obtained the optimal bounds up to logarithmic factors.

The single-buyer case were implicitly studied by Dhangwatnotai et al.~\cite{DRY15} even before Cole and Roughgarden~\cite{CR14} formally introduced the model.
Huang et al.~\cite{HuangMR/2018/SICOMP} gave optimal sample complexity bounds in this setting up to logarithmic factors.

\citet{chen2018brief} studied the targeted query model using a different approach that checks the values at quantiles $\{1,(1+\epsilon)^{-1},(1+\epsilon)^{-2},\ldots\}$. 
Unfortunately, their approach does not generalize to the targeted sampling model to our knowledge. 
This is because they require an $(1+\epsilon)$-multiplicative error of each quantile point, achieving which at the targeted sample model would require too many samples following standard concentration inequalities.

Beyond single-item and more generally single-parameter auctions, the Cole-Roughgarden sample complexity model has also been extensively studied in the multi-parameter setting.
Gonczarowski and Weinberg~\cite{GW18} showed a polynomial upper bound on the sample complexity (information theoretically), despite that a characterization of optimal multi-parameter auctions remains elusive.
Further, Cai and Daskalakis~\cite{CD17} and Balcan et al.~\cite{balcan2018general} analyzed the sample complexity of several families of simple multi-parameter auctions such as item-pricing.
If the buyers' values for different items can be arbitrarily correlated, however, Dughmi et al.~\cite{dughmi2014sampling} proved that exponentially many samples were required.
Recently, Brustle et al.~\cite{brustle2019multi} explored the regime between independent and arbitrarily correlated distributions and proved polynomial sample complexity bounds for certain families of structured correlated distributions.
We leave as a future direction to explore the power of targeted samples in the multi-parameter setting.

\section{Preliminaries}
\label{sec:preliminaries}

\subsection{Model}

Consider a seller selling an item to $\bn \ge 1$ buyers.
The value of each buyer $i$ is drawn independently from a distribution $D_i$.
In other words, the value profile $\bm{v} = (v_1, v_2, \dots, v_n)$ is drawn from a product distribution $\bm{D}=D_1\times D_2\times\cdots\times D_n$.
Without loss of generality, we consider direct revelation mechanisms.
A mechanism consists of an allocation function $\bm{x}$ and a payment function $\bm{p}$.
Let $\bm{b}=(b_1,b_2,\cdots,b_n)$ be the bid profile where each $b_i\ge 0$ is the bid submitted by buyer $i$.
Let $x_i(\bm{b})$ denote the probability that buyer $i$ gets the item;
let $p_i(\bm{b})$ denote buyer $i$'s expected payment.
Each buyer $i$'s utility is quasi-linear, i.e., $v_i\cdot x_i(\bm{b}) - p_i(\bm{b})$.
We focus on dominant strategy incentive compatible (DSIC) and individually rational (IR) mechanisms, i.e., those ensuring that any buyer's utility is nonnegative and maximized if it truthfully reports $b_i = v_i$, regardless of the values and strategies of other buyers.
Hence, the rest of the paper assumes $b_i = v_i$, $1 \le i \le n$.

The seller's objective is to maximize the expected revenue, i.e., the expectation of the sum of all buyers' payments $\sum_{i=1}^{n} p_i(\bm{v})$, where $\vec{v} \sim \vec{D}$.
For any DSIC and IR mechanism $M$ and any product value distribution $\bm{D}$, $\rev{M,\bm{D}}$ denotes the expected revenue of running $M$ on $\bm{D}$.
Let $M_{\bm{D}}$ be the optimal mechanism, which we shall demonstrate shortly.
Define $\opt{\bm{D}} \defeq \rev{M_{\vec{D}},\bm{D}}$.

\subsubsection*{Myerson's Optimal Auction}
Myerson~\cite{Mye81} characterized the revenue optimal auction by introducing the notion of \emph{virtual values}, assuming full knowledge of the prior distribution $\bm{D}$.
Myerson's original definition of virtual values assumes that the distributions have well-defined probability density functions (pdf).
For any buyer $i$, let $F_i$ and $f_i$ denote the cdf and pdf respectively.
Then, the virtual value of buyer $i$, when its value is $v_i$, equals:
\begin{equation}
    \label{eqn:virtual-value}
    \vvi(v_i) \defeq v_i - \frac{1 - F_i(v_i)}{f_i(v_i)}
    ~.
\end{equation}

The definition for arbitrary distributions is more involved.
Since this paper only needs some high-level properties of virtual value rather than its precise definition, we omit the details and refer readers to, e.g., Guo et al.~\cite{GHZ19}.

Myerson observed the following connection between revenue and virtual values:
\begin{equation}
    \label{eqn:revenue-and-virtual-value}
    \E{\sum_{i=1}^n p_i(\vec{v})} = \E{\sum_{i=1}^n \vvi(v_i) x_i(\vec{b})}
    ~.
\end{equation}

The prior $\vec{D}$ is \emph{regular} if $\vvi$ is nondecreasing for any $1 \le i \le n$.
As a special case, the prior has \emph{monotone hazard rate} (MHR) if the second term in Eqn.~\eqref{eqn:virtual-value} is nonincreasing; 
equivalently, it means that the derivative of $\vvi$ is at least $1$.
For regular and MHR distributions, Myerson's optimal auction gives the item to the buyer with the highest \emph{nonnegative} virtual value;
the buyer then pays the lowest bid at which it would still win the item.

In the general case when $\vvi$ may not be monotone, we need another ingredient by Myerson known as \emph{ironing}.
It can be interpreted as identifying an appropriate subset of values $V_i$ for each buyer $i$,%
\footnote{These points span the convex hull of the revenue curve of $D_i$. However, this is unimportant for our results.} 
and constructing a distribution $\bar{D}_i$ by rounding values from $D_i$ down to the closest value in $V_i$.
The resulting $\vec{\bar{D}}$ is regular.
The virtual values defined w.r.t.\ $\vec{\bar{D}}$ are called the \emph{ironed virtual values}, denoted as $\bar{\varphi}_i$'s.
Myerson's optimal auction gives the item to the buyer with the highest nonnegative ironed virtual value.
Importantly, it retains the connection between revenue and virtual value for any mechanism whose decisions depend only on the rounded values (down to the closest one in $V_i$'s):
\begin{equation}
    \label{eqn:revenue-and-ironed-virtual-value}
    \E{\sum_{i=1}^n p_i(\vec{v})} = \E{\sum_{i=1}^n \bar{\varphi}_i(v_i) x_i(\vec{b})}
    ~.
\end{equation}


Finally, the \textit{quantile} of $v_i$ w.r.t.\ buyer $i$'s prior is $q_i(v_i) \defeq \Pr_{v\sim D_i}[v\ge v_i]$.

\subsubsection*{Cole-Roughgarden Model}
Next, suppose we no longer have full knowledge of the prior but instead can access it through i.i.d.\ samples.
This is the model introduced by Cole and Roughgarden~\cite{CR14}, who investigated the number of samples necessary and sufficient for learning a $(1-\epsilon)$-optimal auction.

\subsubsection*{Targeted Sampling Model}
%
This paper considers an alternative model in which the seller can obtain \emph{targeted samples}.
More precisely, the seller may specify a buyer $i$ and a quantile interval in $[0, 1]$ of width at least $\Delta$ and obtain a sample from $D_i$ conditioned on having a quantile within the interval.
A smaller $\Delta$ leads to stronger targeting power, and vice versa.
The Cole-Roughgarden model is a special case when $\Delta = 1$.
Our results demonstrate an advantage compared to the Cole-Roughgarden model whenever $\Delta = o(1)$.
In the limit when $\Delta \to 0$, each query is a pair $(i, q)$ and the answer is precisely the value $v$ whose quantile equals $q$ in buyer $i$'s dataset.
For example, it allows querying the median value of a buyer.
We call this the targeted query model and define targeted query complexity accordingly.

Our analysis is robust to noisy answers to the queries.
Suppose that the queries are answered by a data holder who has learned the distribution through i.i.d.\ samples.
Then, our targeted sample complexity bounds hold as long as the number of i.i.d.\ samples used by the data holder is at least the sample complexity in the Cole-Roughgarden model, i.e., if the data holder has sufficiently many data to learn an approximately optimal auction itself.



\subsection{Technical Preliminaries}

\subsubsection*{Bernstein's Inequality}
This is a standard concentration bound.
We include it below for completeness.
%
\begin{lemma}[Bernstein \cite{bernstein1924modification}]
\label{lem:bernstein}
Let $X_1,X_2,\cdots,X_m$ be i.i.d.\ random variables such that for each $i$, $\E{X_i}=0$, $\E{X_i^2}=\sigma^2$ and $|X_i|\le M$ for some constant $M>0$.
Then for any $t>0$:
\[
\Pr \left[ \sum_{i=1}^{m}X_i > t \right] \le \exp\bigg(-\frac{t^2}{2m\sigma^2+(2/3)Mt}\bigg)
~.\]
~
\end{lemma}

\subsubsection*{Strong Revenue Monotonicity}
A product distribution $\bm{D} = \times_{i=1}^n D_i$ (first-order stochastically) dominates another product distribution $\bm{D}'=\times_{i=1}^nD'_i$, denoted as $\vec{D} \succeq \bm{D'}$,
if for any $1 \le i \le n $ and any value $v$:
\[
    \Pr_{x\sim D_i}[x\ge v]\ge \Pr_{x\sim D'_i}[x\ge v]
    ~.
\]
%
\begin{lemma}[Strong Revenue Monotonicity \cite{DHP16}]
\label{lem:strong-rev-monotone}
Suppose that $\bm{D}$ and $\bm{D}'$ are product distributions, and $\bm{D}\succeq\bm{D}'$.
Recall that $M_{\bm{D}'}$ is the optimal auction for product distribution $\bm{D}'$. 
Then:
\[
\rev{M_{\bm{D}'},\bm{D}}\ge \rev{M_{\bm{D}'},\bm{D}'}
~.\]
\end{lemma}

We also use the following weaker notion of revenue monotonicity, which is a direct corollary.
%
\begin{lemma}[Weak Revenue Monotonicity]
\label{lem:weak-rev-monotone}
If $\bm{D}$ and $\bm{D}'$ are product distributions and $\bm{D}\succeq\bm{D}'$:
\[
\opt{\bm{D}}\ge \opt{\bm{D}'}
~.\]
\end{lemma}

\subsubsection*{Information Theory}
For distributions $P$ and $Q$ on a common sample space $\Omega$, the \textit{Kullback-Leibler (KL) divergence} is: 
\[
D_{KL}(P\|Q) \defeq \int_{\Omega} \ln \bigg(\frac{\df P}{\df Q}\bigg) \df P
~.\]

Further consider a symmetric version:
\[
\dskl{P,Q} \defeq D_{KL}(P\|Q) + D_{KL}(Q\|P)
~.\]

It is known that if $P$ and $Q$ have a small KL divergence then distinguishing them requires a large number of samples.
The next lemma can formalize this connection.
%
\begin{lemma}[e.g., \cite{HuangMR/2018/SICOMP}]
    Suppose an algorithm $A$ distinguishes $P$ and $Q$ correctly with probability at least $\frac{2}{3}$ using $m$ samples. 
    Then, $m$ is at least $\Omega(\dskl{P,Q}^{-1})$.
\end{lemma}

Building on the above lemma, Guo et al.~\cite{GHZ19} further observed that if the KL divergence is small then the optimal revenue is close.
We shall let the parameter $\alpha$ in the next lemma be either $\epsilon$ for additive approximation results, and $\epsilon \cdot \opt{\vec{D}}$ for multiplicative results.
%
\begin{lemma}[\cite{GHZ19}]
\label{lem:kl-rev}
If product distributions $\bm{D}$ and $\bm{D}'$ satisfy that for some $K>0$, some $\alpha>0$ and some sufficiently small constant $c>0$:
\begin{enumerate}
    \item They have small KL divergence: $\dskl{\bm{D},\bm{D}'}\le cK^{-1}$.
    \item For any mechanism, and either of the two distributions, $K$ samples are sufficient to estimate the expected revenue up to an additive $\alpha$ error with probability at least $\frac{2}{3}$.
\end{enumerate}
Then we have:
\[
    \opt{\bm{D}}\ge \opt{\bm{D}'}-2\alpha
    ~.
\]
\end{lemma}

\subsubsection*{Yao's Minimax Principle}
Our lower bounds use \emph{Yao's minimax principle} to handle randomized algorithms.
For simplicity, we state below only its instantiation on our problem rather than the general form.
%

\begin{lemma}[Yao's Minimax Principle \cite{yao1977probabilistic}]
\label{lem:Yao}
Let $\mathcal{D}$ be any family of priors and $\mathcal{M}$ be the family of DSIC and IR mechanisms.
Then, for any distributions $\mu_M$ over $\mathcal{M}$ and $\mu_D$ over priors in $\mathcal{D}$:
\[
    \max_{\vec{D} \in \mathcal{D}} \textsc{\textbf{E}}_{M \sim \mu_M} \big[ \opt{\vec{D}} - \rev{M, \vec{D}} \big]
    \ge
    \min_{M \in \mathcal{M}} \textsc{\textbf{E}}_{\vec{D} \sim \mu_D} \big[ \opt{\vec{D}} - \rev{M, \vec{D}} \big]
    ~.
\]
\end{lemma}


\section{Multiple Buyers: Key Ingredients and Targeted Query Model}
\label{sec:multi}

In this section we propose our three key ingredients, and illustrate them in the proof of targeted query complexity upper bounds for learning near optimal multi-buyer single-item auctions. 
The generalized case of targeted sample complexity upper bounds, i.e., when $0 < \Delta < 1$, are deferred to Section~\ref{app:large-delta}.

\begin{theorem}
    \label{thm:main}
    For any $0 < \epsilon < 1$ and any number of buyers $n$, there is an algorithm that learns an auction with expected revenue at least $(1-\epsilon) \opt{\bm{D}}$:
    \begin{enumerate}
        \item with $\tilde{O} (\ep^{-3/2})$ targeted queries if $\bm{D}$ is regular; or
        \item with $\tilde{O} (\ep^{-1})$ targeted queries if $\bm{D}$ is MHR; or
        \item with $\tilde{O} (H^{1/2}\ep^{-1})$ targeted queries if $\bm{D}$ is $[1,H]$-bounded.
    \end{enumerate}
    The algorithm also learns an auction with expected revenue at least $\opt{\bm{D}}-\ep$:
    \begin{enumerate}
        \setcounter{enumi}{3}
        \item with $\tilde{O} (\ep^{-1})$ targeted queries if $\bm{D}$ is $[0,1]$-bounded.
    \end{enumerate}
\end{theorem}

The algorithm is nonadaptive and agnostic to the underlying family of distributions. 
As sketched in the introduction, the proof relies on three technical ingredients, which will be developed in Subsections~\ref{sec:theta}, \ref{sec:sandwich}, and \ref{sec:pinpoint} respectively.
As thought experiments that provide context and motivation of the lemmas, these subsections include several informal and suboptimal arguments of targeted sample complexity, which eager reader may skip.
Finally, Subsection~\ref{sec:multi-proof} presents formal proof.




\subsection{Ingredient 1: Very Few (Top) Samples Matter}
\label{sec:theta}

This subsection establishes the observation that most i.i.d.\ samples are not useful for the task of learning a nearly optimal auction.
Recall the examples on the two extremes.
On the one hand, suppose the buyers are close to i.i.d.\ in the sense that they make similar contribution to the optimal revenue.
Then, with high probability there is at least one buyer whose value falls into the top $\tilde{O} \big( \frac{1}{n} \big)$ quantile of its prior.
It is not difficult to show that focusing only on such high-value buyers gets most of the optimal revenue.
To learn the correct decisions w.r.t.\ these high-value buyers, only a $\tilde{O} \big( \frac{1}{n} \big)$ portion of the samples per buyer is useful.
On the other hand, if there is only one buyer who may have a high value and the other buyers are negligible, every sample for the former buyer matters.
Nonetheless, we may still argue that most i.i.d.\ samples, in particular those for the negligible buyers, do not matter.

Next we formally show that only a $\tilde{O} \big( \frac{1}{n} \big)$ portion of the i.i.d.\ samples matter in general.
To state this mathematically, we need some notations.
For any threshold $0 \le \theta \le 1$ in the quantile space, and any distribution $D$, consider a truncated distribution obtained by rounding values with quantile greater than $\theta$ down to $0$.
Denote the truncated distribution as $D_\theta$.
In other words, the quantile of $D_\theta$ is defined as:
\[
    q^{D_\theta}(v) 
    \defeq
    \begin{cases}
        \min\{q^{D}(v), \theta\} & \text{if } v>0 ~;\\
        1 & \text{if } v=0 ~.
\end{cases}
\]

Further, for any product distribution $\bm{D}$ and any vector of thresholds $\bm{\theta}$ let $\vec{D}_{\vec{\theta}}$ be the product distribution whose $i$-th coordinate is obtained by truncating $D_i$ at threshold $\theta_i$.

\begin{lemma}
    \label{lem:theta}
    For any $\epsilon$ and any product value distribution $\bm{D}$, there is a threshold vector $\bm{\theta}$ so that:
    \begin{enumerate}[label={(\arabic*)}]
        \item $\sum_{1\le i\le n}\theta_i \le \log \frac{1}{\epsilon} + 1$;
        \item $\opt{ \vec{D}_{\vec{\theta}} } \ge (1-\epsilon)\opt{\bm{D}}$.
    \end{enumerate}
    The latter implies $\opt{ \vec{D}_{\vec{\theta}} }\ge \opt{\bm{D}} - \epsilon$ for $[0,1]$-bounded distributions.
\end{lemma}

\begin{proof}
    Motivated by the connection between ironed virtual values and revenue, we select the vector of thresholds $\vec{\theta}$ based on an appropriate threshold $\varphi^*$ in the ironed virtual value space.
    Then, we define $\theta_i$ for every buyer $i$ to be a quantile at which the ironed virtual value equals $\varphi^*$, with appropriate tie-breaking if there is an interval of values/quantiles all having ironed virtual value $\varphi^*$.
    Intuitively, $\varphi^*$ shall be sufficiently large in order to satisfy (1), and shall be sufficiently small to retain the majority of the optimal revenue and satisfy (2).

    To formalize the translation of a threshold in the ironed virtual value space to thresholds in the quantile space, consider the following notations.
    For any $\varphi^*$, and any buyer $i$, define:
    \begin{align*}
        \overline \theta_i( \varphi^* ) & \defeq \sup \big\{ q_i(v_i) : \bar{\varphi}_i(v_i) \ge \varphi^* \big\} ~; \\
        \underline \theta_i( \varphi^* ) & \defeq \inf \big\{ q_i(v_i) :\bar{\varphi}_i(v_i) \le \varphi^* \big\} ~.
    \end{align*}

    Then, by definition we have:
    \begin{equation}
        \label{eqn:few-value-matter-threshold}
        \bar{\varphi}_i(v_i) 
        \begin{cases}
            > \varphi^* & \text{ if } q_i(v_i) < \underline{\theta}_i(\varphi^*) ~; \\
            = \varphi^* & \text{ if } \underline{\theta}_i(\varphi^*) < q_i(v_i) < \overline{\theta}_i(\varphi^*) ~; \\
            < \varphi^* & \text{ if } q_i(v_i) > \overline{\theta}_i(\varphi^*) ~.
        \end{cases}
    \end{equation}
    The second case is relevant only when $ \overline \theta_i( \varphi^* ) \ne \underline \theta_i( \varphi^* )$.
    The boundary cases when $q_i(v_i) = \underline{\theta}_i(\varphi^*)$ or $\overline{\theta}_i(\varphi^*)$ are ambiguous;
    they happen with probability $0$ and, thus, do not affect the analysis.

    Importantly, Eqn.~\eqref{eqn:few-value-matter-threshold} implies that the values at $\overline{\theta}_i$ and $\underline{\theta}_i$ are not in the interior of ironed intervals.
    Hence, we can apply the connection between revenue and ironed virtual values, i.e., Eqn.~\eqref{eqn:revenue-and-ironed-virtual-value} in the following argument.

    Next, we characterize the largest possible $\varphi^*$ using the following conditions on $\varphi^*$ and $\theta_i$'s that are sufficient for getting (2) in the lemma:
    \begin{enumerate}
        \item[(3)] For any buyer $i$, $\theta_i \in \big\{ \overline{\theta}_i(\varphi^*), \underline{\theta}_i(\varphi^*) \big\}$; and
        \item[(4)] $\prod_{i=1}^n (1-\theta_i) \le \epsilon$.
    \end{enumerate}

    \paragraph{(3) (4) $\Rightarrow$ (2).}
    The case when $\varphi^* \le 0$ is trivial since truncating negative ironed virtual values does not decrease the optimal revenue at all.

    Next assume $\varphi^* > 0$.
    The probability that no bidder $i$ has a value $v_i$ above its corresponding threshold, i.e., $q_i(v_i) > \theta_i$, for $1 \le i \le n$, is $\prod_{i=1}^n (1-\theta_i) \le \epsilon$ (condition (4) above).
    Observe that whenever there exist some bidders with values above their thresholds, Myerson's optimal auction gives the item to one of them.
    Hence, by the connection between revenue of ironed virtual values, the contribution by values above the thresholds is at least $(1 - \epsilon) \varphi^*$ (condition (4) and first two cases of Eqn.~\eqref{eqn:few-value-matter-threshold});
    on the other hand, the contribution of value below the thresholds is at most $\epsilon \varphi^*$ (condition (4) and last two cases of Eqn.~\eqref{eqn:few-value-matter-threshold}).
    Putting together implies (2) in the lemma.

    \paragraph{Satisfying (1), (3), and (4).}
    It remains to choose $\varphi^*$ and $\theta_i$'s appropriately to satisfy not only (3) and (4) (which implies (2)) but also (1).
    We start with the choice of $\varphi^*$:
    \[
        \varphi^* = \max \bigg\{ \varphi : \prod_{i=1}^n \big(1 - \overline{\theta}_i(\varphi) \big) \le \epsilon \bigg\}
        ~.
    \]

    We can take the maximum here because $\overline{\theta}_i(\varphi)$ is nonincreasing and left-continuous in $\varphi$.
    Then, taking $\theta_i = \overline{\theta}_i(\varphi)$ for all $1 \le i \le n$ satisfies (3) and (4) but may violate (1).
    Next we change $\theta_i$ from $\overline{\theta}_i(\varphi^*)$ to $\underline{\theta}_i(\varphi^*)$ (so that (3) is still satisfied) for an appropriate subset of $i$ to recover (1) without violating (4).
    A simple greedy strategy suffices:
    \begin{enumerate}
        \item[(i)] Start with $\theta_i = \overline{\theta}_i(\varphi^*)$ for all $1 \le i \le n$.
        \item[(ii)] For $i = 1, 2, \dots, n$, change $\theta_i$ to $\underline{\theta}_i(\varphi)$ if doing so still satisfies (4) with strict inequality.
        \item[(iii)] Return the final $\vec{\theta}$.
    \end{enumerate}

    We first show that the above will not change all $\theta_i$ to $\underline{\theta}_i(\varphi^*)$.
    To see this, let $\varphi \downarrow \varphi^*$ from above:
    \[
        \prod_{i=1}^n \big( 1 - \underline{\theta}_i(\varphi^*) \big) = \lim_{\varphi \downarrow \varphi^*} \prod_{i=1}^n \big( 1 - \overline{\theta}_i(\varphi) \big) \ge \epsilon
        ~.
    \]

    Next, suppose we let $\theta_i = \underline{\theta}_i(\varphi^*)$ for an extra $i$.
    By definition:
    \begin{equation}
        \label{eqn:theta-lemma}
        \prod_{i=1}^n \big( 1 - \theta_i \big) \ge \epsilon
        ~.
    \end{equation}

    Then, we have:
    \[
        \sum_{1\le i\le n} \theta_i =n-\sum_{1\le i\le n}(1-\theta_i) \le n \big(1-{\epsilon}^{\frac1n}\big) \le n \log \epsilon^{-\frac{1}{n}} = \log \epsilon^{-1} ~.
    \]
    Here, the first inequality follows from the inequality of arithmetic and geometric means and Eqn.~\eqref{eqn:theta-lemma}.
    Condition (1) follows because letting $\theta_i = \underline{\theta}_i(\varphi)$ for an extra $i$ changes the sum by at most $1$.
\end{proof}

\paragraph{Thought Experiment: Known Thresholds.}
To motivate the next ingredient, suppose that we know not only the existence but also the values of the thresholds in Lemma~\ref{lem:theta}.
Consider $[0, 1]$-bounded distributions as the running example.
For any distribution $D$, an application of the argument of Guo et al.~\cite{GHZ19} allows us to learn from $N$ samples a dominated empirical distribution $\tilde{E}$ sandwiched between $D$ and another distribution $\tilde{D}$, i.e., $D \succeq \tilde{E} \succeq \tilde{D}$, and the KL divergence between $D$ and $\tilde{D}$ is $\tilde{O} \big( \frac{1}{N} \big)$.

For every buyer $i$, let $D$ in the above argument be the distribution of $v_i \sim D_i$ conditioned on $q_i(v_i) \le \theta_i$.
By targeting the top $\theta_i$ portion of each buyer $i$'s prior with $N$ targeted samples, the same argument allows us to learn a dominated distribution $\tilde{E}$ sandwiched between $D$ and the corresponding $\tilde{D}$.
Next, consider three distributions whose top $\theta_i$ portion are $D$, $\tilde{E}$, and $\tilde{D}$ above;
denote them as $D_{\theta_i}$, $\tilde{E}_{\theta_i}$, and $\tilde{D}_{\theta_i}$ respectively.
We remark that $D_{\theta_i}$ is precisely the $i$-th coordinate of $\vec{D}_{\vec{\theta}}$, the truncated version of $\vec{D}$.
By $D_{SKL}\big(D, \tilde{D}\big) = \tilde{O}\big(\frac{1}{N}\big)$, we have:
\begin{equation}
    \label{eqn:dskl-condition}
    \dskl{D_{\theta_i}, \tilde{D}_{\theta_i}} = \tilde{O} \left( \frac{\theta_i}{N} \right)
    ~.
\end{equation}

Then, let:
\[
    \vec{D}_{\vec{\theta}} = \times_{i=1}^n D_{\theta_i} ~, \quad
    \vec{\tilde{E}}_{\vec{\theta}} = \times_{i=1}^n \tilde{E}_{\theta_i} ~, \quad
    \vec{\tilde{D}}_{\vec{\theta}} = \times_{i=1}^n \tilde{D}_{\theta_i} ~.
\]

We have:
\begin{enumerate}[label={(\arabic*)}]
    \item $\opt{\vec{D}_{\vec{\theta}}} \ge \opt{\vec{D}} - \epsilon$;
        \hspace*{\fill} (Lemma~\ref{lem:theta})
    \item $\vec{D} \succeq \vec{D}_{\vec{\theta}} \succeq \vec{\tilde{E}}_{\vec{\theta}} \succeq \vec{\tilde{D}}_{\vec{\theta}}$; and
    \item $D_{SKL} \big( \vec{D}_{\vec{\theta}}, \vec{\tilde{D}}_{\vec{\theta}} \big) = \sum_{i=1}^n D_{SKL} \big( D_{\theta_i}, \tilde{D}_{\theta_i} \big) = \tilde{O} \left( \frac{1}{N} \right)$.
        \hspace*{\fill} (Eqn.~\eqref{eqn:dskl-condition} and Lemma~\ref{lem:theta})
\end{enumerate}

Therefore, the revenue of Myerson's optimal auction for $\vec{\tilde{E}}_{\vec{\theta}}$ on the true distribution $\vec{D}$ is at least the optimal revenue of $\vec{\tilde{E}}_{\vec{\theta}}$ by strong revenue monotonicity. 
This is at least the optimal revenue of $\vec{\tilde{D}}_{\vec{\theta}}$ by weak revenue monotonicity.
This in turn is at least the optimal revenue of $\vec{D}_{\vec{\theta}}$ minus $\epsilon$ with $N = \tilde{O}\big(\epsilon^{-2}\big)$, by condition (3) above and Lemma~\ref{lem:kl-rev}.
The approximation guarantee now follows by condition (1) above.

\subsection{Ingredient 2: Sandwich Lemma}
\label{sec:sandwich}

This subsection extends the algorithm and its argument sketched in the previous subsection to handle unknown thresholds.
The argument in the subsection is suboptimal compared to the final analysis.
We present it to motivate the construction of the shading function to be shown as Eqn.~\eqref{eqn:shade-function}, and later restated as Eqn.~\eqref{eqn:shade-function-restate}.
Eager readers may skip this subsection.

The algorithm uses a doubling trick.
For every buyer $i$ it makes $\log n$ guesses of $\theta_i = 1, \frac{1}{2}, \frac{1}{4}, \dots, \frac{1}{n}$;
for every guess of $\theta_i$ it targeted samples $\tilde{O}(N^2)$ times from quantile interval $[0, \theta_i]$ where $N$ is a parameter to be determined.
Readers may think of $N$ as the number of targeted queries stated in Theorem~\ref{thm:main}.
To this end, this subsection proves a weaker bound of $\tilde{O}(N^2)$;
the final quadratic improvement will be covered in the next subsection.
The doubling trick can be intuitively viewed as getting more samples in the large value (small quantile) regime, and fewer samples in the small value (large quantile) regime.

To simplify the argument, we targeted sample from the symmetric quantile intervals $[1 - \theta_i, 1]$ as well for every guess of $\theta_i$.
The results still hold without these samples.

\paragraph{(Aggregated) Empirical Distribution}
Define the empirical distribution $E_i$ of each buyer $i$ by aggregating the targeted samples for all guesses of $\theta_i$.
Concretely, estimate its conditional distribution in quantile interval $[0, \frac{1}{n}]$ using targeted samples with $\theta_i = \frac{1}{n}$, and in quantile interval $[2^{-i}, 2^{-i+1}]$ using targeted samples with $\theta_i = 2^{-i+1}$ for any $2 \le i \le \log n$.
Estimate the other half of the distribution with quantile intervals in $[\frac{1}{2}, 1]$ symmetrically.

By Bernstein's inequality (Lemma~\ref{lem:bernstein}), for any value with quantile $q$, the above doubling trick estimates its quantile up to an additive error of (with high probability):
\[
    \begin{cases}
        \frac{1}{N} \cdot \sqrt{\frac{q}{n}} \quad & q \in [0, \frac{1}{n}] ~; \\
        \frac{1}{N} \cdot q & q \in (\frac{1}{n}, \frac{1}{2}] ~.
    \end{cases}
\]
and also with the symmetric error bounds for quantiles in $(\frac{1}{2}, 1]$.
Applying union bound on all quantiles that are multiples of $\frac{1}{nN^2}$, with high probability the quantile estimation is accurate everywhere up to the above error bound plus $\frac{1}{nN^2}$.

\paragraph{Shading Function.}
Driven by the above estimation error upper bounds, consider the following function $f(\cdot)$ defined on the quantile space $[0,1]$:
\begin{equation}
\label{eqn:shade-function}
    f(q)=
    \begin{cases}
        \frac{1}{N} \cdot \sqrt{\frac{q}{n}} + \frac{1}{N^2n} \quad & q\in [0,\frac{1}{n}] ~; \\
        \frac{1}{N} \cdot q + \frac{1}{N^2n} & q\in(\frac{1}{n},\frac{1}{2}] ~; \\[.5ex]
        f(1-q) & q\in(\frac{1}{2},1] ~.
    \end{cases}
\end{equation}

\begin{lemma}
    \label{lem:doubling-trick-error}
    With high probability, the empirical distributions from the doubling trick satisfy that for any buyer $i$ and any value $v_i$:
    \[
        \big| q^{E_i}(v_i) - q^{D_i}(v_i) \big| \le f \big( q^{D_i}(v_i) \big)
        ~.
    \]
\end{lemma}

Further define the following shading functions based on $f$:
\[
    \shades(q) \defeq \max \big \{ 0, q - f(q) \big\} ~,\quad
    \shaded(q):=\max \big\{ 0, q - 2f(q) \big\} ~.
\]

For any distribution $D$, we abuse notation and also let $\shades(D)$ denote a distribution such that for any value $v>0$, $q^{\shades(D)}(v):=\shades(q^D(v))$.
Let $\shades(\bm{D})$ denote the product distribution obtained by applying $\shades$ to each coordinate of $\bm{D}$.
Define $\shaded(D)$ and $\shaded(\vec{D})$ similarly.
The next lemma implies that the shaded distributions are well defined.
The proof follows by basic calculus and is thus omitted.

\begin{lemma}
    \label{lem:monotone}
    Both $\shades$ and $\shaded$ are nondecreasing in $q \in [0, 1]$.
\end{lemma}

Next, define the dominated empirical distribution $\vec{\tilde{E}}$ to be $\shades(\vec{E})$;
further define $\vec{\tilde{D}} = \shaded(\vec{D})$.
By the definitions of the distributions and Lemma~\ref{lem:doubling-trick-error} we have:

\begin{lemma}
    \label{lem:D-tildeE-tildeD}
    $\vec{D} \succeq \vec{\tilde{E}} \succeq \vec{\tilde{D}}$.
\end{lemma}

Finally, we arrive at the main lemma of the subsection.

\begin{lemma}[Sandwich Lemma]
    \label{lem:dskl}
    For any buyer $i$ and any threshold $\theta_i$ in the quantile space and the corresponding value $v_i$ w.r.t.\ $D_i$, define distributions $D_{\theta_i}$, $\tilde{E}_{\theta_i}$, and $\tilde{D}_{\theta_i}$ by rounding values smaller than $v_i$ down to $0$ in $D_i$, $\tilde{E}_i$, and $\tilde{D}_i$ respectively.
    Then:
    \begin{enumerate}
        \item $D_{\theta_i} \succeq \tilde{E}_{\theta_i} \succeq \tilde{D}_{\theta_i}$; and
    \end{enumerate}
    Suppose further that $D_i$'s is defined on $[0, u]$ with point masses of at least $\frac{9}{N^2 n}$ at $u$ and $\ell = 0$.
    \begin{enumerate}
        \setcounter{enumi}{1}
    \item $D_{SKL} \big( D_{\theta_i},\tilde{D}_{\theta_i} \big) = \tilde O \big( \frac{\theta_i}{N^2} \big) + O \big( \frac{1}{N^2n} \big)$.
    \end{enumerate}
\end{lemma}

\begin{proof}
    The first part follows by the definitions of the distributions and Lemma~\ref{lem:D-tildeE-tildeD}.

    Next consider the second part.
    Let the probability masses of point mass $u$ and $\ell=0$ be $p_u$ and $p_{\ell}$ respectively.
    If $\theta_i < p_u$, $D_{\theta_i}$ is a point mass at $0$ and the lemma is trivially true.
    If $1 - p_\ell < \theta_i < 1$, the resulting $D_{\theta_i}$ is identical to the case when $\theta_i = 1 - p_\ell$.
    If $\theta_i = 1$, it degenerates to the counterpart in Guo et al.~\cite{GHZ19}.
    Hence, it suffices to consider the case when $p_u \le \theta_i \le 1 - p_\ell$ below.

    First assume there are no point masses other than $0$ and $u$.
    Then the KL divergence equals: 
    %
    \begin{equation}
        \label{eqn:sandwich}
    \begin{aligned}
        \dskl{D_{\theta_i},\tilde{D}_{\theta_i}}
        & ~=~ \big(p_u-\shaded(p_u)\big)\ln\frac{p_u}{\shaded(p_u)} & \textrm{(point mass at $u$)} \\[1ex]
        & \qquad + \big( 1-\theta_i-\shaded(1-\theta_i) \big) \ln\frac{1-\shaded(1-\theta_i)}{\theta_i} & \textrm{(point mass at $0$)} \\
        & \qquad + \int_{0<v<u} \left( \frac{d\tilde D_{\theta_i}}{dD_{\theta_i}}-1 \right) \ln \frac{d\tilde D_{\theta_i}}{dD_{\theta_i}} dD_{\theta_i} &
    \end{aligned}
    \end{equation}

    \paragraph{Point mass at $u$.}
    By $\ln(1+x) \le x$, the first term is bounded by:
    \begin{equation}
        \label{eqn:mass-at-u}
        (p_u-\shaded(p_u))\ln\frac{p_u}{\shaded(p_u)} = 2f(p_u)\ln\frac{p_u}{p_u-2f(p_u)}\le \frac{4f(p_u)^2}{p_u-2f(p_u)}
        ~.
    \end{equation}

    For ease of notations, let $a \defeq \frac{1}{N\sqrt{n}}$ in the following argument.
    For any $p_u \le \frac{1}{n}$, by definition:
    \[
        f(p_u) = a \sqrt{p_u} + a^2
        ~.
    \]

    Further by the assumption $p_u \ge \frac{9}{N^2 n} = 9a^2$, we have:
    \[
        f(p_u) \le \frac{4}{9} p_u
        ~.
    \]

    Hence, the denominator of the RHS of Eqn.~\eqref{eqn:mass-at-u} is at least $\Omega(p_u)$.
    Furthe observe that $f(p_u) = O(a \sqrt{p_u})$ (again due to $p_u \ge 9 a^2$).
    The RHS of Eqn.~\eqref{eqn:mass-at-u} is therefore at most:
    \[
        O \left( \frac{f(p_u)^2}{p_u} \right) = O(a^2)
        ~.
    \]

    For any $\frac{1}{n} < p_u \le \frac{1}{2}$, by definition and the assumption of $p_u \ge 9a^2$:
    \[
        f(p_u) = \frac{p_u}{N}+a^2 \le \big( \frac{1}{N} + \frac{1}{9} \big) p_u
        ~.
    \]

    Hence, the denominator of the RHS of Eqn.~\eqref{eqn:mass-at-u} is again at least $\Omega(p_u)$.
    Therefore, the RHS of Eqn.~\eqref{eqn:mass-at-u} is upper bounded by:
    \[
        O \bigg( \frac{f(p_u)^2}{p_u} \bigg)
        \le
        O \bigg( \frac{(p_u/N)^2}{p_u} + \frac{(a^2)^2}{p_u} \bigg)
    \]

    The first term is at most $\frac{p_u}{N^2} \le \frac{\theta_i}{N^2}$.
    The second term is at at most $O(a^2)$ due to $p_u \ge 9a^2$.

    Finally, since $f(p_u)$ takes its maximum value at $p_u = \frac{1}{2}$, the case of $p_u = \frac{1}{2}$ further implies the desired bound for any $p_u > \frac{1}{2}$.

    \paragraph{Point mass at $\ell = 0$.}
    By $\ln(1+x)\le x$, the second term of Eqn.~\eqref{eqn:sandwich} is at most:
    \[
        \big(1-\theta_i-\shaded(1-\theta_i)\big) \ln\frac{1-\shaded(1-\theta_i)}{\theta_i}
        =
        2f(1-\theta_i)\ln\bigg(1+\frac{1-\theta_i-\shaded(1-\theta_i)}{\theta_i}\bigg)
        \le
        \frac{4f(1-\theta_i)^2}{\theta_i}
        ~.
    \]

    By symmetry of $f$ the above is $\frac{4 f(\theta_i)^2}{\theta_i}$.
    The rest of the proof is identical to the previous case, replacing $p_u$ with $\theta_i$.

    \paragraph{Interior.}
    The third term of Eqn.~\eqref{eqn:sandwich} can be written as:%
    \footnote{Although $f(q)$ and thus $d_f$ is indifferentiable at $q=\frac{1}{n}$ and $q=\frac{1}{2}$, the Radon–Nikodym derivative is well defined viewing $d_f$ as a measure function. In particular, setting $d_f'(\frac{1}{n})$ and $d_f'(\frac{1}{2})$ arbitrarily does not affect the calculation.}
    %
    \begin{equation}
        \label{eqn:interior}
        \int_{p_u<q\le \theta_i}(\shaded'(q)-1)\ln \shaded'(q) dq
        ~.
    \end{equation}

    By the definition of $\shaded$:
    \[
        |\shaded'(q)-1| = |2f'(q)|=
        \begin{cases}
            \sqrt{\frac{1}{N^2nq}}, & 0 \le q < \frac{1}{n} ~; \\[1ex]
            \frac{2}{N} & \frac{1}{n} < q < \frac{1}{2}  ~; \\[1ex]
            2|f'(1-q)| & \frac{1}{2} < q \le 1 ~.
        \end{cases}
    \]

    Since any $q$ in Eqn.~\eqref{eqn:interior} is between $p_u$ and $1 - p_\ell$ and further $p_u, p_\ell \ge \frac{9}{N^2n}$, we have $|f'(q)|\le \frac{1}{2}$. 
    By $x\ln(1+x)\le 2x^2$ for $|x| \le 1/2$, Eqn.~\eqref{eqn:interior} is upper bounded by $2$ times:
    \[
        \int_{p_u}^{\theta_i} \big( f'(q) \big)^2 dq
        ~.
    \]

    If $\theta_i \le \frac{1}{n}$, this is (recall $p_u \ge \frac{9}{N^2 n})$:
    \[
        \int_{p_u}^{\theta_i} \frac{1}{N^2 n q} dq = \frac{\ln \frac{\theta_i}{p_u}}{N^2 n} \le O \bigg( \frac{ \log N }{N^2 n} \bigg)
        ~.
    \]

    If $\frac{1}{n} < \theta_i \le \frac{1}{2}$, the contribution from $p_u$ to $\frac{1}{n}$ is bounded by $\tilde{O}(\frac{1}{N^2 n})$ due to the previous case.
    It remains to bound the integration from $\frac{1}{n}$ to $\theta_i$:
    \[
        \int_{\frac{1}{n}}^{\theta_i} \frac{4}{N^2} dq
        = O \bigg(\frac{\theta_i}{N^2} \bigg)
        ~.
    \]

    \paragraph{Point masses.}
    For any point mass that corresponds to some quantile interval $(a, b]$, its contribution to the KL divergence is bounded by the contribution of $q \in (a, b]$ in the above calculation due to Jensen's inequality. 
    This is identical to \citet{GHZ19}, so we omit the details.
\end{proof}



Modulo the assumption of having mild point masses at the endpoints of the distributions' supports, which can be handled using standard technique, we have recovered all technical properties needed in the thought experiment in the previous subsection simultaneously for all possible thresholds, in particular for the vector of thresholds in Lemma~\ref{lem:theta}.
Hence, by the same argument in the thought experiment in the beginning of this section, $\tilde{O} \big( \epsilon^{-2} \big)$ targeted queries (or targeted samples with $\Delta \le \frac{1}{n}$) suffice.
We shall not repeat the argument in details here since an even better bound is applicable using a new technical ingredient, which will be covered in the next subsection.

\subsection{Ingredient 3: Pinpoint Lemma}
\label{sec:pinpoint}

The argument in the last subsection does not rely on the details of the doubling trick.
Instead, it only requires the resulting dominated empirical distribution $\vec{\tilde{E}}$ is sandwiched between the true distribution $\vec{D}$ and the auxiliary distribution $\vec{\tilde{D}} = \shaded(\vec{D})$.
To be self-contained, we restate the definition of $f$ in Eqn.~\eqref{eqn:shade-function} below:
\begin{equation}
    \label{eqn:shade-function-restate}
    f(q) =
    \begin{cases}
        \frac{1}{N} \cdot \sqrt{\frac{q}{n}} + \frac{1}{N^2n} \quad & q\in [0,\frac{1}{n}] ~; \\
        \frac{1}{N} \cdot q + \frac{1}{N^2n} & q\in(\frac{1}{n},\frac{1}{2}] ~; \\[.5ex]
        f(1-q) & q\in(\frac{1}{2},1] ~.
    \end{cases}
\end{equation}

We also restate the definition of $\shaded$ below:
\[
    \shaded(q) \defeq \max \big\{ 0, q - 2f(q) \big\} ~.
\]

We abuse notation and let $\shaded(\vec{D})$ be the distribution obtained by applying $\shaded$ to the quantile of every value in every coordinate of $\vec{D}$.

Since the fact that $\vec{\tilde{E}}$ is obtained via the doubling trick is unimportant, this subsection directly designs an algorithm that achieves the same sandwiching property using fewer targeted queries than the doubling trick.
In other words, we will abandon the doubling trick even though the choice of the auxiliary distribution $\vec{\tilde{D}}$ and the subsequent analysis is motivated by it.

Define the following pinpoints recursively:
\begin{equation}
    \label{eqn:pinpionts}
    \begin{aligned}
        q_0 & = 1 ~; \\
        q_{j+1} & = \shaded(q_j) = q_j - 2f(q_j) \qquad \text{(whenever $\shaded(q_j) > 0$)} ~.
    \end{aligned}
\end{equation}

Let $k$ denote the index of the last term of the sequence.
Further write $q_{k+1} = 0$ for convenience.

Our targeted querying algorithm directly constructs the dominated empirical distributions $\vec{\tilde{E}}$ by querying the pinpoints.
See Algorithm~\ref{alg:shade-algorithm} below.

\begin{algorithm}
	\caption{Dominated Empirical Myerson Auction from Targeted Queries}
	\label{alg:shade-algorithm}
	\begin{algorithmic}[1]
	\FOR{each buyer $1 \le i \le n$}
        \STATE Targeted query quantiles $q_j$ for $1 \le j \le k$; let $v_j$'s be the corresponding values.
        \STATE Further write $v_0 = 0$.
        \STATE Let $\tilde{E}_i$ be a discrete distribution with a point mass $q_j - q_{j+1}$ at $v_j$, $0 \le j \le k$.
	\ENDFOR
    \STATE \textbf{output:} Myerson's optimal auction w.r.t. $\bm{\tilde{E}}$.
	\end{algorithmic}
\end{algorithm}

\begin{lemma}[Pinpoint Lemma, Part 1]
    \label{lem:distribution-domination}
    The dominated empirical distribution constructed in Algorithm~\ref{alg:shade-algorithm} satisfies that:
    \[
        \bm{D} \succeq \bm{\tilde{E}}\succeq \tilde{\bm{D}} ~.
    \]
\end{lemma}

\begin{proof}
    In other words, we need to show that for any value $v \ge 0$, its quantiles w.r.t.\ the three distributions satisfy:
    \[
        q^{D_i}(v) \ge q^{\tilde{E}_i}(v) \ge q^{\tilde{D}_i}(v)
        ~.
    \]

    It holds trivially with equality when $v = 0$.

    Next, assume $v > 0$.
    Write $v_{k+1} = \infty$ for convenience.
    Then, for every value $v$ there exists $0 \le j \le N$ such that $v_j < v \le v_{j+1}$.
    Then, its quantile w.r.t.\ $\tilde{E}_i$ is by definition equal to $q_{j+1}$, i.e.:
    \[
        q^{\tilde{E}_i}(v) = q_{j+1}
        ~.
    \]
    
    Further, its quantile w.r.t.\ $D_i$ is at least the quantile of $v_{j+1}$ w.r.t.\ $D_i$ which equals $q_{j+1}$, i.e.:
    \[
        q^{D_i}(v) \ge q^{D_i}(v_{j+1}) = q_{j+1}
        ~.
    \]
    
    Finally, its quantile w.r.t.\ $\tilde{D}_i$ is at most the quantile of $v_j$ w.r.t.\ $\tilde{D}_i$, which by definition equals $\shaded(q_j) = q_{j+1}$, i.e.:
    \[
        q^{\tilde{D}_i}(v) \le q^{\tilde{D}_i}(v_j) = \shaded\big(q^{D_i}(v_j)\big) = \shaded(q_j) = q_{j+1}
        ~.
    \]

    Putting together proves the lemma.
\end{proof}

It remains to upper bound the number of pinpoints.

\begin{lemma}[Pinpoint Lemma, Part 2]
\label{lem:pinpoint-2}
The number of pinpoints $k$ is at most $\tilde{O}(N)$.
\end{lemma}

\begin{proof}
    First consider the number of pinpoints $[\frac{1}{n},\frac{1}{2}]$.
    (The same argument applies to pinpoints in $[\frac{1}{2}, 1 - \frac{1}{n}]$.)
    By the definition of $\shaded$ and the corresponding shaded function $f$ in Eqn.~\eqref{eqn:shade-function}, we have:
    \[
        q_{j+1} 
        =
        q_j - \tfrac{2}{N} \cdot q_j - \tfrac{2}{N^2 n}
        <
        \left(1-\tfrac{2}{N}\right) q_j
        ~.
    \]

    Hence, there are at most $\log_{1+\frac{2}{N}} \frac{n}{2} = \tilde{O}(N)$ pinpoints in $[\frac{1}{n},\frac{1}{2}]$.
    .

    Next consider pinpoints in $[0, \frac{1}{n})$. (The same argument applies to pinpoints in $(1 - \frac{1}{n}, 1]$.)
    By definition we have:
    \[
        q_{j+1}
        =
        q_j - \tfrac{2}{N} \cdot \sqrt{\tfrac{q_j}{n}} - \tfrac{2}{N^2 n}
        <
        \left( \sqrt{q_j} - \tfrac{1}{N\sqrt{n}} \right)^2
        ~.
    \]

    Taking square-root on both sides gives:
    \[
        \sqrt{q_{j+1}} \le \sqrt{q_j} - \tfrac{1}{N\sqrt{n}}
        ~.
    \]

    Hence, there at most $N$ pinpoints with between $[0, \frac{1}{n})$ since there square-roots are in $[0, \frac{1}{\sqrt{n}})$.
\end{proof}

\subsection{Proof of Theorem~\ref{thm:main}}
\label{sec:multi-proof}

We start by introducing several truncation operators on distributions to so that the resulting distributions having point masses at the endpoints of the their supports, as required by Lemma~\ref{lem:dskl}. 

Let $\trvi{D_i}$ be the distribution obtained by truncating values larger than $\bar{v}_i$ down to $\bar{v}_i$. 
In other words, the quantile of the resulting distribution is defined as:
\[
    q^{\trvi{D_i}}(v) \defeq 
    \begin{cases}
    0 & \text{if } v>\bar{v}_i ~;\\
    q^{D_i}(v) & \text{if } v\le \bar{v}_i ~.
\end{cases}
\]

We allow it to be applied to product distribution coordinate-wise, written as:
\[
    \trv{\bm{D}} = \times_{i=1}^n \trvi{D_i}~.
\]

Further define $\treps{q}$ given any $\ep\in[0,1]$, which truncates the lowest $\ep$ fraction of values of a distribution $D$ down to $0$.
That is, the quantile of the resulting distribution is defined as:
\[
    q^{\treps{D}}(v) \defeq
    \begin{cases}
        \min\{q^D(v),1-\ep\} & \text{if } v>0 ~; \\
        1 & \text{if } v=0 ~.
    \end{cases}
\]

We also allow it to be applied to product distribution coordinate-wise, written as:
\[
    \treps{\bm{D}} = \times_{i=1}^n \treps{D_i}~.
\]
\begin{lemma}[Lemma 11 in~\cite{GHZ19}]
\label{lem:lem11}
For any product value distribution $\bm{D}$, we have:
\[
    \opt{t^{\min}_{\ep}(\bm{D})}\ge (1-\ep)\opt{\bm{D}}
    ~.
\]
\end{lemma}

\begin{lemma}
\label{lem:lem12}
For any product value distribution $\bm{D}$, suppose that there are $\bar{\bm{v}}$, $\beta\ge \opt{\bm{D}}$, and $p>0$ such that:
\begin{enumerate}[label={(\arabic*)}]
\item $\beta\ge p\bar{v}_i$, $\forall i\in[n]$.
\item $q^{D_i}(\bar{v}_i)\ge p\epsilon^2n^{-1}$.
\item $\opt{\trv{\bm{D}}}\ge \opt{\bm{D}}-\epsilon\beta$.
\end{enumerate}
Let $\bm{\theta}$ be the threshold vector in Lemma~\ref{lem:theta}. 
There are value distributions $\bm{D}'_{\bm{\theta}}$ and $\tilde{\bm{D}}'_{\bm{\theta}}$ such that:
\begin{enumerate}
\item[(a)]$\bm{D}'_{\bm{\theta}}$ and $\tilde{\bm{D}}'_{\bm{\theta}}$ have bounded support in $[0,\frac{\beta}{p}]$.
\item[(b)]$\opt{\bm{D}'_{\bm{\theta}}}\ge \opt{\bm{D}}-3\epsilon \beta$.
\item[(c)]$\tilde{\bm{D}}\succeq \tilde{\bm{D}}'_{\bm{\theta}}$. 
\item[(d)]$D_{SKL}(\tilde{\bm{D}}'_{\bm{\theta}},\bm{D}'_{\bm{\theta}})=O(p\epsilon^2)$.
\end{enumerate}
\end{lemma}
\begin{proof}
Set $\bm{D}':=t^{\min}_{\epsilon}(\trv{\bm{D}})$. 
Then we can define $\bm{D}'_{\bm{\theta}}$ and $\tilde{\bm{D}}'_{\bm{\theta}}=\shaded(\bm{D}'_{\bm{\theta}})$.
We verify the four results one by one. 

\paragraph{Part (a).}
It holds by the definition of $\bm{D}'=t^{\min}_{\epsilon}(\trv{\bm{D}})$ and condition (1) in the lemma.

\paragraph{Part (b).} 
By condition (3) in this lemma, and weak revenue monotonicity (Lemma~\ref{lem:weak-rev-monotone}), we have:
\[
\opt{\bm{D}} \ge \opt{\trv{\bm{D}}} \ge \opt{\bm{D}}-\beta \ep 
~.\]

Further by $\bm{D}'_{\bm{\theta}}\preceq t^{\min}_{\ep}(\trv{\bm{D}}) \preceq \trv{\bm{D}}\preceq \bm{D}$, we have:
\begin{align*}
    \opt{\trv{\bm{D}}} 
    &\ge \opt{\bm{D}'_{\bm{\theta}}} & \textrm{(weak revenue monotonicity, i.e. Lemma~\ref{lem:weak-rev-monotone})}\\
    &\ge (1-\ep)\opt{t^{\min}_{\ep}(\trv{\bm{D}}} & \textrm{(Lemma~\ref{lem:theta})}\\
    &\ge (1-\ep)^2\opt{\trv{\bm{D}}} & \textrm{(Lemma~\ref{lem:lem11})}\\
    &\ge \opt{\trv{\bm{D}}}-2\ep \opt{\bm{D}} & \textrm{(weak revenue monotonicity, i.e. Lemma~\ref{lem:weak-rev-monotone})}\\
    &\ge  \opt{\trv{\bm{D}}}-2\ep\beta &
\end{align*}
Therefore by condition (3) in this lemma:
\[
\opt{\bm{D}} \ge \opt{\bm{D}'_{\bm{\theta}}} \ge \opt{\trv{\bm{D}}}-2\ep\beta \ge \opt{\bm{D}}-3\beta \ep 
~.\]

\paragraph{Part (c).}
Since $\tilde{\vec{D}} = \shaded(\vec{D})$, $\tilde{\vec{D}}_{\vec{\theta}}' = \shaded(\vec{D}_{\vec{\theta}}')$, it holds by $\vec{D}\succeq \vec{D}_{\vec{\theta}}'$, and the monotonicity of $\shaded(\cdot)$.

\paragraph{Part (d).}
By our choice of $\bm{D}'_{\bm{\theta}}$ and $\tilde{\bm{D}}'_{\bm{\theta}}$, we apply Lemma~\ref{lem:dskl} with distribution $\bm{D}'$, $u=\bar{v_i}$, $ \Omega(\frac{1}{N^2n})=p_u\ge \max\{p\ep^2n^{-1},\ep\} $. 
Choose $N$ with proper log factors and we have:
\[
\dskl{D'_{\theta_i},\tilde D'_{\theta_i} }
= \tilde O\left(\frac{\theta_i}{N^2}+\frac{1}{N^2n}\right) 
= O\left(\frac{p\ep^2}{n}+\theta_ip\ep^2\right)
~,\]
which implies:
\[
\dskl{\tilde{\bm{D}}'_{\bm{\theta}},\bm{D}'_{\bm{\theta}}} = O\left(\sum_i \left[\frac{p\ep^2}{n}+\theta_ip\ep^2\right]\right) = O\left(p\ep^2\right)
~.\]
\end{proof}

\begin{lemma}[\cite{GHZ19}]
\label{lem:beta-p}
There exists $\bar{\bm{v}}$ satisfying the conditions in Lemma~\ref{lem:lem12} if:
\begin{enumerate}
    \item $\bm{D}$ is $[1,H]$-bounded, and set $\beta = \opt{\bm{D}}$, $p = \frac{1}{H}$; or
    \item $\bm{D}$ is regular, and set $\beta = \opt{\bm{D}}$, $p = \frac{\ep}{8}$; or
    \item $\bm{D}$ is MHR, and set $\beta = \opt{\bm{D}}$, $p = \Theta(\frac{1}{\log(2/\ep)})$; or
    \item $\bm{D}$ is $[0,1]$-bounded, and set $\beta =1$, $p = 1$.
\end{enumerate}
\end{lemma}

\begin{proof}[Proof of Theorem~\ref{thm:main}]
    Below we choose $\beta$ and $p$ according to Lemma~\ref{lem:beta-p}, and apply Lemma~\ref{lem:lem12}.
Let $\bm{D}':=t^{\min}_{\epsilon}(\trv{\bm{D}})$.
Let $\bm{\tilde{E}}$ be the (dominated) empirical distribution by Algorithm~\ref{alg:shade-algorithm}. 
By Lemma~\ref{lem:distribution-domination}, $\bm{D}\succeq\vec{\tilde{E}}\succeq \tilde{\bm{D}}$. 
Further, Part (c) of Lemma~\ref{lem:lem12} implies $\vec{\tilde{E}}\succeq\tilde{\bm{D}}'_{\bm{\theta}}$. 
Then, we have:
\begin{align*}
    \rev{M_{\bm{E}},\bm{D}}
&\ge \opt{\bm{E}} & \textrm{(strong revenue monotonicity, i.e. Lemma~\ref{lem:strong-rev-monotone})}\\
&\ge \opt{\tilde{\bm{D}}'_{\bm{\theta}}} & \textrm{(weak revenue monotonicity, i.e. Lemma~\ref{lem:weak-rev-monotone})}\\
&\ge \opt{\bm{D}'_{\bm{\theta}}}-2\ep\beta & \textrm{(part (a) and (d) in Lemma~\ref{lem:lem12}, and Lemma~\ref{lem:kl-rev})}\\[1ex]
&\ge \opt{\bm{D}}-5\ep\beta ~.& \textrm{(part (b) in Lemma~\ref{lem:lem12})}
\end{align*}

By condition (2) in Lemma~\ref{lem:lem12}, and the condition in Lemma~\ref{lem:dskl}, 
we can have $N=\tilde{\Theta}(p^{-1/2}\ep^{-1})$. 
Further, note that we can reduce the factor of $5$ by multiplying constant factors to $N$. 
So we get the results of Theorem~\ref{thm:main} by considering the $\beta$'s and $p$'s for different distribution families.
\end{proof}

\section{Multi-buyer: Targeted Sampling Model}
\label{app:large-delta}

In this section we consider the case of $0<\Delta<1$, 
further recall that we only care about $\Delta$ which is not sufficiently small, so that it won't lead to a degeneration to targeted query model. 
As discussed in Section~\ref{sec:sandwich}, the doubling trick in fact shows a feasible algorithm for $\Delta\le \frac{1}{n}$. 
We first proceed to show that for the case $\Delta>\frac{1}{n}$, we can also use the doubling trick and will obtain non-trivial results. 

In addition, when we draw one targeted sample from quantile interval $[q,q+\Delta]$, the error of the sampled value's quantile is bounded by $\Delta$.
Thus, when $\Delta$ is close to $0$ the case almost becomes targeted query. 
We will further show a feasible algorithm for $\Delta\le \frac{1}{n}$ in Section~\ref{sec:app-small-Delta}. 

More formally, 
we consider the set of all sampled intervals in the following form:
\[
[0,\Delta],[\Delta,2\Delta],[2\Delta,4\Delta],\cdots,[\frac{1}{4},\frac{1}{2}],[\frac{1}{2},\frac{3}{4}],\cdots,[1-\Delta,1]
~.\]
In each interval, assume that we draw $N$ targeted samples. 
We can construct an (aggregated) empirical distribution $\bm{E}$ with all these samples. 
The algorithm described in Section~\ref{sec:sandwich} is a special case where $\Delta=\frac{1}{n}$. 
The total number of targeted samples is $\tilde{O}(N)$, since there are totally $\tilde{\Theta}(1)$ intervals. 

We follow again the routine in Section~\ref{sec:app-small-Delta}.  
First, we design proper shading functions $\shades$ and $\shaded$ based on Bernstein's inequality. 
The concentration bound (Lemma~\ref{lem:app-bernstein}) will be proved in Section~\ref{sec:app-bernstein}. 
Then we show that for $\tilde{\bm{E}}:=\shades(\bm{E})$, $\tilde{\bm{D}}:=\shaded(\bm{D})$ we have $\bm{D}\succeq\tilde{\bm{E}}\succeq\tilde{\bm{D}}$ in Lemma~\ref{lem:app-monotone}. 

Given the above, we are ready to propose the sandwich lemma for $\Delta\ge \frac{1}{n}$ in the next lemma.
\begin{lemma}[Sandwich Lemma]
\label{lem:app-sandwich}
    For any buyer $i$ and any threshold $\theta_i$ in the quantile space and the corresponding value $v_i$, define distributions $D_{\theta_i}$, $\tilde{E}_{\theta_i}$, and $\tilde{D}_{\theta_i}$ by rounding values smaller than $v_i$ down to $0$ in $D_i$, $\tilde{E}_i$, and $\tilde{D}_i$ respectively.
    Then:
    \begin{enumerate}
        \item $D_{\theta_i} \succeq \tilde{E}_{\theta_i} \succeq \tilde{D}_{\theta_i}$; and
    \end{enumerate}
    
    Suppose further that $D_i$'s is defined on $[0, u]$ with point masses of $\tilde \Omega(\frac{\Delta}{N})$ at $u$ and $\ell = 0$.
    Then:
    \begin{enumerate}
        \setcounter{enumi}{1}
        \item $\dskl{D_{\theta_i},\tilde{D}_{\theta_i}} = \tilde O (\frac{\theta_i}{N}+\frac{\Delta}{N})$.
    \end{enumerate}
\end{lemma}

Finally our main result of this part is summerized in the next theorem. 
Given Lemma~\ref{lem:app-sandwich}, we can determine $N$, and the proof is almost identical as that in Section~\ref{sec:multi-proof} so we will omit the proof.
\begin{theorem}
\label{thm:app-main}
Suppose $\Delta \ge  \frac{1}{n}$. 
For any $\epsilon\in(0,1)$ and any $n$-bidder product value distribution $\bm{D}$, with probability $1-\delta$, 
there exists an algorithm that can return an expected revenue at least $(1-\epsilon)\opt{\bm{D}}$:
\begin{enumerate}
    \item with $\tilde{O} (n\Delta\ep^{-3})$ targeted samples if $\bm{D}$ is regular; or
    \item with $\tilde{O} (n\Delta\ep^{-2})$ targeted samples if $\bm{D}$ is MHR; or
    \item with $\tilde{O} (n\Delta H\ep^{-2})$ targeted samples if $\bm{D}$ is $[1,H]$-bounded.\\
    There also exists an algorithm that can return an expected revenue at least $\opt{\bm{D}}-\ep$:
    \item with $\tilde{O} (n\Delta\ep^{-2})$ targeted samples if $\bm{D}$ is $[0,1]$-bounded.
\end{enumerate}
\end{theorem}

\subsection{Extended concentration bound}
\label{sec:app-bernstein}
We extend the concentration bound to our case, and the analysis here is based on a directly application of Bernstein inequality, as shown in the next lemma. 
\begin{lemma}\label{lem:bstn2}
Suppose we have sampled $N$ values from $D$ in quantile interval $[a,b]\subseteq [0,1]$ and construct an empirical distribution $E$, we will have for any $v$ with quantile w.r.t. $D$ in $[a,b]$, with probability $1-\delta$:
\[
|q^E(v)-q^D(v)|\le \sqrt{\frac{2(q^D(v)-a)(b-q^D(v))\cdot L}{N}} + \frac{L}{N}(b-a)
~,\]
where $L=O(\log (N/\delta))$.
\end{lemma}
\begin{proof}
Consider quantiles that are multiples of $\frac{b-a}{N}$ in $[a,b]$ and the corresponding values. 
We first net the whole quantile space w.r.t. $D$ into pieces of length $\frac{b-a}{N}$.
For such value $v$ and its quantile w.r.t. $D$, $q^D(v)$, we construct the $N$ variables: $Y={y_1, y_2, \cdots, y_N}$ where $y_i=(\mathbb{I}[v_i\ge v]-q^D(v))/N$. We have $\mathbb{E}[Y] =0$ and $\sigma^2_Y = (b-q^D(v))(q^D(v)-a)/N^2$. Then by the Bernstein inequality with $\vert y_i\vert <M=\frac{b-a}{N}$,
\begin{align*}
    Pr\left[\left\vert \sum_i y_i \right \vert > t\right]&\le
    2exp\left(\frac{-t^2}{2N\cdot \sigma_Y^2+2/3\cdot M \cdot t}\right)\\
    &= 2exp\left(\frac{-t^2}{2N\cdot (b-q^D(v))(q^D(v)-a)/N^2+2/3\cdot \dfrac{b-a}{N}\cdot t}\right).\\
\end{align*}
We need the right term smaller than $\frac{\delta}{N}$. 
By choosing a proper $L=O(\log (N/\delta))$, we can obtain 
\begin{align*}
    t &\le \dfrac{\dfrac{\frac{2}{3}(b-a) L}{N}+\sqrt{\left(\dfrac{\frac{2}{3}(b-a) L}{N}\right)^2+ L\cdot \dfrac{8}{N}\cdot (b-q^D(v))(q^D(v)-a)}}{2}\\
    &\le \sqrt{\frac{2(q^D(v)-a)(b-q^D(v))\cdot L}{N}} + \frac{2L(b-a)}{3N}.
\end{align*}
So for any sampled value from the quantile interval $[a,b]$, its quantile should be located between two adjacent quantiles with a distance of $\frac{b-a}{N}$.
We get $|q^E(v)-q^D(v)|$ is upper bounded by:
\[
\sqrt{\dfrac{2(q^D(v)-a)(b-q^D(v))L}{N}} + \dfrac{2L}{3N}(b-a) + \dfrac{b-a}{N}
\le\sqrt{\dfrac{2(q^D(v)-a)(b-q^D(v))L}{N}} + \dfrac{L}{N}(b-a)
~.\]
\end{proof}

We can directly extend this lemma to product value $\bm{D}$, and consider many sampled intervals at the same time.
Suppose all sampled intervals are $\{[a_j,b_j]\}$, and we targeted sample $N_j$ times from interval $[a_j,b_j]$. 
We denote $I_{min}=\min_j\{(b_j-a_j)/N_j\}$, and net the whole quantile space w.r.t. each $D_i$ into intervals of length $I_{min}$. 
Then the points at the multiples of $I_{min}$ should all be close to the true quantile with probability $1-\delta\cdot I_{min}/n$. 
Adjusting the value of $L$ and therefore we have the following lemma:
\begin{lemma}
\label{lem:app-bernstein}
With probability $1-\delta$, for any value $v$, for any buyer $i\in[n]$, if $q^{D_i}(v)\in[a_j,b_j]$ we have:
\[
|q^{E_i}(v)-q^{D_i}(v)| \le \sqrt{\frac{2(q^{D_i}(v)-a_j)(b_j-q^{D_i}(v))L}{N_j}} + \frac{L(b_j-a_j)}{N_j}
~,\]
where $L=O(\log {\frac{n}{\delta I_{min}}})$ .
\end{lemma}

\subsection{Shading functions and KL-divergence. }
In the following analysis, we design proper shading functions based on Lemma~\ref{lem:app-bernstein}. 
The shading function should be monotone and differentiable\footnote{Also the following definition of $f(q)$ is not differentiable. As discussed in Section~\ref{sec:multi} we omit the tiny adjustment and assume this $f(q)$ is proper.}, 
and the shaded amount should cover the bound in Lemma~\ref{lem:app-bernstein}. 
Also, let $L$ be the one used in Lemma~\ref{lem:app-bernstein}. 
\begin{equation}
\label{eqn:app-f}
f(q)=\left\{
    \begin{aligned}
    &2\sqrt{\frac{qL\Delta}{N}}, & q\in [0,\Delta]\\
    &2\sqrt{\frac{L}{N}}\cdot q & q\in(\Delta,\frac{1}{2}]\\
    &f(1-q) & q\in(\frac{1}{2},1]
    \end{aligned}
\right.
\end{equation}

Further define shading functions based on $f(q)$:
\[
    \shades(q):=\max\{ 0,q-f(q)-\frac{4L\Delta}{N}\}~, \shaded(q):=\max\{ 0,q-2f(q)-\frac{5L\Delta}{N}\}~,
\]
and let $\tilde{\bm{D}}=\shaded(\bm{D}),\shades(\bm{E})=\tilde{\bm{E}}$ by applying the function on quantile of each dimension. 

\begin{lemma}
\label{lem:app-monotone}
$\shades$ and $\shaded$ are monotone increasing in $q$. 
Therefore, $\tilde{\bm{E}}:=\shades(\bm{E})$ and $\tilde{\bm{D}}:=\shaded(\bm{D})$ are well defined.
\end{lemma}
\begin{proof}
First, consider $\shades(q)$. 
Solve $q=f(q)+\frac{4L\Delta}{N}$ we have $q=\Theta(\frac{L\Delta}{N})$. 
Then we can choose some $p_u=\Theta(\frac{L\Delta}{N})$ and claim that for any $q\in (p_u,1-p_u)$, $\shades(q)$ is positive, and otherwise $\shades(q)=0$. 
Calculate $f'(q)$ on the following three intervals respectively: 
\begin{equation}
\label{eqn:app-f'}
f'(q)=\left\{
    \begin{aligned}
    &\sqrt{\frac{L\Delta}{Nq}}, & q\in [p_u,\Delta)\\
    &2\sqrt{\frac{L}{N}} & q\in(\Delta,\frac{1}{2})\\
    &-f'(1-q) & q\in(\frac{1}{2},1-p_u]
    \end{aligned}
\right.
\end{equation}
We can checked that $f'(q)\le 1$ case by case.
\begin{enumerate}
\item $q\in [p_u,\frac{1}{n})$. $f'(q)\le f'(p_u)=\sqrt{\frac{L\Delta}{p_uN}}$. By a proper $p_u=\Theta(\frac{L\Delta}{N})$ we can have $f'(p_u)\le 1$.
\item $q\in(\frac{1}{n},\frac{1}{2})$. It is true because $\frac{L}{N}=o(1)$.
\item $q\in(\frac{1}{2},1-p_u]$. In this case $f'(q)\le 0$.
\end{enumerate}
The proof of $\shaded(q)$ is almost identical. 
Therefore, $\tilde{\bm{E}}:=\shades(\bm{E})$ and $\tilde{\bm{D}}:=\shaded(\bm{D})$ are well defined.
\end{proof}
\begin{lemma}
\label{lem:app-distribution-domination}
Assuming the bounds in Lemma~\ref{lem:app-bernstein}, and our choice of $f(q)$ in Eqn.~(\ref{eqn:app-f}), we have:
\[
    \bm{D}\succeq \tilde{\bm{E}}\succeq \tilde{\bm{D}}~,
\]
where $\tilde{\bm{E}}:=\shades(\bm{E})$ and $\tilde{\bm{D}}:=\shaded(\bm{D})$.
\end{lemma}

See Appendix~\ref{app:proof-of-domination} for a detailed proof of Lemma~\ref{lem:app-distribution-domination}.
\bigskip

\begin{proof}[Proof of Lemma~\ref{lem:app-sandwich}]
The first part follows by the definitions of the dsitributions and the previous lemma. 
It remains to prove the second part.

Let the probability masses of point mass $u$ and $\ell$ be $p_u$ and $p_{\ell}$ respectively.
First assume there are no point masses other than $0$ and $u$, and consider $\theta_i\in(p_u,1-p_{\ell})$. 
For $\theta_i\notin (p_u,1-p_{\ell})$, the calculation is almost identical, and we can prove by modifying end points.  
Then the KL divergence can be written as:
\begin{equation*}
\begin{aligned}
    \dskl{D'_{\theta_i},\tilde{D}'_{\theta_i}} = \bigg( &(p_u-\shaded(p_u))\ln\frac{p_u}{\shaded(p_u)}\bigg) & \textrm{(Point mass at $u$.)}\\
                                 &+ \bigg((1-\theta_i-\shaded(1-\theta_i))\ln\frac{1-\shaded(1-\theta_i)}{\theta_i} \bigg)& \textrm{(Point mass at $0$.)}\\
                                 &+ \int_{0<v<u}\bigg( \ln \frac{d\tilde D'_{\theta_i}}{dD'_{\theta_i}}(\frac{d\tilde D'_{\theta_i}}{dD'_{\theta_i}}-1)\bigg)dD'_{\theta_i} &
\end{aligned}
\end{equation*}

By almost an identical calculation, as in Lemma~\ref{lem:dskl}, the first term and the second term is $\tilde{O}\left(\frac{\theta_i}{N}+\frac{\Delta}{N}\right)$. 
Third,  we calculate:
\[
\int_{p_u<q\le \theta_i}(\shaded'(q)-1)\ln(\shaded'(q))dq
~.\]
\begin{equation*}
|d'(q)-1| = |2f'(q)|=\left\{
    \begin{aligned}
    &2\sqrt{\frac{L\Delta}{Nq}}, & q\in [p_u,\Delta)\land q\le \theta_i\\
    &4\sqrt{\frac{L}{N}} & q\in(\Delta,\frac{1}{2})\land q\le \theta_i\\
    &2|f'(1-q)| & q\in(\frac{1}{2},1-p_{\ell}]\land q\le \theta_i
    \end{aligned}
\right.
\end{equation*}
By our choice of proper $p_u,p_{\ell}$, we can have $|f'(q)|\le \frac{1}{2}$. 
By $x\ln(1+x)\le 2x^2$ for $|x|\le 1/2$, the third term is then upper bounded by $\int_{p_u}^{\theta_i}2|f'(q)|^2dq$. 
Also note that $|f'(q)|=|f'(1-q)|$. 
Calculate $|2f'(q)|$ in each part and we have:
\begin{align*}
\int_{p_u<q\le \theta_i}(\shaded'(q)-1)\ln(\shaded'(q))dq
&\le \int_{p_u<q\le \theta_i}2|f'(q)|^2 dq\\
&= \tilde O (\frac{\theta_i}{N}) ~.
\end{align*}
If it has point masses other than $0$ and $u$, the argument is identical as that in \cite{GHZ19}.

To sum up, we have
\[
    D_{SKL}(D'_{\theta_i},\tilde{D}'_{\theta_i})=\tilde O (\frac{\theta_i}{N})+ \tilde O(\frac{\Delta}{N})~.
\]

\end{proof}

\subsection{The case of $\Delta<\frac{1}{n}$.}
\label{sec:app-small-Delta}

We will now consider the case of $\Delta\le \frac{1}{n}$\footnote{To avoid a tiny $\Delta$ which will lead to a degeneration to targeted query model, it suffices to only consider $\Delta\ge \frac{1}{N^2n}$. }.  
Briefly speaking, we will conbine the ideas from $\Delta\ge \frac{1}{n} $ , and $\Delta=0$.
We will use the same $f(q)$ in Section~\ref{sec:multi-proof} which will ensure a small enough KL-divergence. 
\begin{equation}
\label{eqn:app-shade-function}
    f(q)=
    \begin{cases}
        \frac{1}{N} \cdot \sqrt{\frac{q}{n}} + \frac{1}{N^2n}, & q\in [0,\frac{1}{n}]\\
        \frac{1}{N} \cdot q + \frac{1}{N^2n} & q\in(\frac{1}{n},\frac{1}{2}]\\
        f(1-q) & q\in(\frac{1}{2},1]
    \end{cases}
\end{equation}

Then, our algorithm is descriped in Algorithm~\ref{alg:shade-algorithm-appendix}. 
Briefly speaking, in each round the algorithm calculates how many samples is needed to ensure that:
the concentration bound induced by Lemma~\ref{lem:app-bernstein}, is at most $f(q)$. 
If $f(q)< \Delta$
then we need more than one targeted sample, the algorithm will behave similarly to case of $\Delta\ge \frac{1}{n}$. 
Our analysis is then based on concentration bound as in the previous section. 
Otherwise, the algorithm will in fact, do a targeted query with an additive error of $\Delta$. 
Then the analysis is based on Pinpoint Lemma, as discussed in Section~\ref{sec:multi}. 

More precisely, if the needed number of targeted samples is less than one, then we are ready to apply Pinpoint Lemma, by using Algorithm~\ref{alg:shade-algorithm}. 
This in fact implies $\Delta$ is so tiny that $f(q)\ge \Delta$. 
Note that this is equivalently a targeted query with an additive error of $\Delta$.
Therefore, we may define shading functions based on $f(q)$:
\[
    \shades(q):=\max\big\{ 0,q-\min\{f(q),\Delta\}\big\}~, 
    \shaded(q):=\max\{ 0,q-2f(q)\}~,
\]

\begin{algorithm}[!h]
	\caption{Dominated Empirical Distribution for $\Delta<\frac{1}{n}$.}
	\label{alg:shade-algorithm-appendix}
	\begin{algorithmic}[1]
	\STATE \textbf{input:} parameters $n,\Delta,\{N_j\}$, shading function $\shades$ and function $f$.
	\STATE For each buyer $i\in[n]$, 
	\FOR{$i=1,2,\ldots,n$}
		\STATE Starting from $j=1$, $a_0=0$ we targeted sample $N_j$ times from interval $[a_{j-1},a_{j-1}+\Delta]$ in round $j$. 
		\STATE Set $a_j=j\Delta$ and $j=j+1$. 
		\STATE Stop if $a_j=\frac{1}{2}$, or $f(a_j)\ge \Delta$. 
		\STATE If the latter happens, apply Algorithm~\ref{alg:shade-algorithm} on interval $[a_{j},1-a_{j}]$, set the step to be $f$ instead of $2f$. 
		That is, starting from $q_0=1-a_j$, 
		we replace targeted querying quantile $q_k$ by targeted sampling once in quantile interval $[q_k-\Delta, q_k]$, 
		and set $q_{k+1}=q_k-f(q_k)$, whenever the right-hand side is at least $a_j$. 
		\STATE Repeat the above in a symmetric way in $[1-a_j,1]$. 
		\STATE Construct the empirical distribution $E_i$ according to the above probability.
	\ENDFOR
	\STATE Use $\shades$ to construct the shade empirical $\tilde{\bm{E}}=\tilde{E}_1\times \tilde{E}_2\times \ldots\times \tilde{E}_n$ as follows:\\
		\begin{equation*}
		q^{\tilde{E}_i}(v)=
			\begin{cases}
			\max\big\{ 0,q^{E_i}(v)-\min\{f(q^{E_i}(v)),\Delta\}\big\}& \text{if $v>0$}\\
			1& \text{if $v=0$}
			\end{cases}
		\end{equation*}
	\STATE \textbf{output:} Myerson's optimal auction w.r.t. $\tilde{\bm{E}}$.
	\end{algorithmic}
\end{algorithm}

For proper parameters, we can summarize our results of this part as the next theorem. 
The proof is similar as that in the case of $\Delta\ge \frac{1}{n}$.
We include the choice of parameters and the proof in Appendix~\ref{app:proof-of-larger-Delta}. 
\begin{theorem}
\label{thm:all-Delta}
    For any $\epsilon \in (0,1)$, any $\Delta\in (0,1)$ and any number of buyers $n$, there is an algorithm that learns an auction with expected revenue at least $(1-\epsilon) \opt{\bm{D}}$:
    \begin{enumerate}
        \item with $\tilde{O} (\max\{n\Delta\ep^{-3},\ep^{-1.5}\})$ targeted samples if $\bm{D}$ is regular; or
        \item with $\tilde{O} (\max\{n\Delta\ep^{-2},\ep^{-1}\})$ targeted samples if $\bm{D}$ is MHR; or
        \item with $\tilde{O} (\max\{n\Delta H\ep^{-2},\sqrt{H}\ep^{-1}\})$ targeted samples if $\bm{D}$ is $[1,H]$-bounded.
    \end{enumerate}
    There is also an algorithm that learns an auction with expected revenue at least $\opt{\bm{D}}-\ep$:
    \begin{enumerate}
        \setcounter{enumi}{3}
        \item with $\tilde{O} (\max\{n\Delta\ep^{-2},\ep^{-1}\})$ targeted samples if $\bm{D}$ is $[0,1]$-bounded.
    \end{enumerate}
\end{theorem}

\section{Single Buyer}\label{sec:single}
In the single buyer case, we again consider four families of distributions: Regular, $[0, 1]$-bounded, $[1, H]$-bounded and MHR. We first focus on the case of query complexity upper bounds. After that, we give the lower bounds of query complexity to almost match the upper bound.  

\subsection{Upper Bounds of Targeted Query Complexity}

We sketch the proofs of the theorems below and the missing proofs can be found in Appendix \ref{prf:upper_single}.

 \begin{theorem} \label{thm:single_upper_all}
     For any $\epsilon \in (0,1)$ and a single buyer, there is an algorithm that learns an auction with expected revenue at least $(1-\epsilon) \opt{D}$:
    \begin{enumerate}
        \item with $O(\log \epsilon^{-1})=\tilde{O}(1)$ targeted queries if $D$ is {regular} or {MHR}; or
        \item with $O(\epsilon^{-1}\cdot \log H)=\tilde{O}(\epsilon^{-1})$ targeted queries if $D$ is {$[1,H]$}{-bounded}.
    \end{enumerate}
    The algorithm also learns an auction with expected revenue at least $\opt{D}-\epsilon$:  
    \begin{enumerate}
        \setcounter{enumi}{2}
        \item with $\tilde{O} (\ep^{-1})$ targeted queries if $D$ is $[0,1]$-bounded.
    \end{enumerate}
 \end{theorem}





\subsubsection*{Proof Sketch} 
It is known that the optimal mechanism for a single buyer is to set a take-it-or-leave-it price at the highest point of the revenue-quantile ($R-q$) curve. 
Consider the $R-q$ curve of the buyer, each time we query a quantile, we will know the value at this quantile and calculate the revenue at this point. Since for regular and MHR distributions, the $R-q$ curve is concave, it is natural to apply the idea of binary search. Each time we query the whole interval equidistantly and select the sub-interval with the largest possible revenues to make our target interval shrink for a constant ratio. For $[0,1]$ and $[1,H]$ bounded distributions, we just net the whole quantile interval into many sub-intervals.
At the endpoints of these intervals, the revenue is approximately accurate. 
Since the revenues of any quantile is close to the revenue of its adjacent ends, we only need to choose the maximum revenue from the queried points.

This idea can be extended to cases of targeted samples, with the help of concentration bound. 
If $\Delta$ isn't sufficiently small, we also prove the following theorem in Appendix \ref{prf:upper_single}.
\begin{theorem}
\label{thm:single_upper_sample}
For any $\epsilon \in (0,1)$ and a single buyer, there is an algorithm that learns an auction with expected revenue at least $(1-\epsilon) \opt{D}$:
    \begin{enumerate}
        \item with $\tilde{O}(\max \{1,\min\{\Delta^2\epsilon^{-4}, \Delta \epsilon^{-3}\}\})$ targeted samples if $D$ is {regular}; or
        \item with $\tilde{O}(\max \{1,\Delta^2\epsilon^{-2}\})$ targeted samples if $D$ is {MHR}; or
        \item with $\tilde{O}(\max\{H^{1/2}\epsilon^{-1},\Delta H\epsilon^{-2}\})$ targeted samples if $D$ is {$[1,H]$}{-bounded}.
    \end{enumerate}
    The algorithm also learns an auction with expected revenue at least $\opt{D}-\epsilon$:  
    \begin{enumerate}
        \setcounter{enumi}{3}
        \item with $\tilde{O}(\max\{\epsilon^{-1}, \Delta\epsilon^{-2}\})$ targeted samples if $D$ is $[0,1]$-bounded.
    \end{enumerate}
\end{theorem}

\subsection{Lower Bounds of Targeted Query Complexity}

Finally, we give lower bounds for different families of distributions matching the upper bounds in the previous subsection.
The upper and lower bounds match up to a constant in the case of regular, $[1,H]$-bounded and $[0,1]$ bounded distributions.
For MHR distributions, we do not know if the logarithmic factor in the upper bound is necessary.
Nonetheless, the upper bound is tight up to a single logarithmic factor since we need at least $1$ query.
We sketch the proof below, deferring the formal proofs to Appendix \ref{prf:lower_single}.

\begin{theorem}\label{thm:single_lower_all}
     For any $\epsilon \in (0,1)$ and a single buyer, any algorithm should make at least:
    \begin{enumerate}
        \item $\Omega(\log(\epsilon^{-1}))$ targeted queries to achieve at least $(1-\epsilon)\opt{D}$ expected revenue if the buyer's distribution is {regular}; or
        \item
        $\Omega(\epsilon^{-1})$ targeted queries to achieve at least $\opt{D}-\epsilon$ expected revenue if the buyer's distribution is {$[0,1]$}{-bounded}; or
        \item
        $\Omega(\epsilon^{-1}\cdot\log H)$ targeted queries to achieve at least $(1-\epsilon)\opt{D}$ expected revenue if the buyer's distribution is {$[1,H]$}{-bounded}.
    \end{enumerate}
 \end{theorem}

\subsubsection*{Proof Sketch}
Note that the lower bound of targeted query complexity should also hold for that of the targeted sample complexity, because to targeted sample once, we can turn it into uniformly draw a quantile in the sampling interval and query at this quantile. We only show the query complexity for the three cases. The proofs of these three cases for query complexity lower bound rely on Yao's Minimax Principle(Lemma~\ref{lem:Yao}). 
We first turn the expected loss for any randomized algorithm on a deterministic distribution into the expected loss for a deterministic algorithm on a randomized distribution. 
We need to build a set of value distributions (i.e. $R-q$ curves) so that no deterministic algorithm is able to return expected revenue loss less than $\epsilon$ on a randomized distribution over this distribution set. 
For regular cases, we build this set iteratively with the idea of the binary search, while keeping all the distributions concave. 
For $[0,1]$-bounded and $[1,H]$-bounded distributions, we construct a set of distribution: they have almost the same revenue-quantile curve with each other, except for a tiny quantile interval.  
If the algorithm fails to distinguish between two distributions, there would be a $2\epsilon$ revenue loss. Since knowing the locations of the differences to rule out possibilities requires certain number of samples, we then can calculate the lower bound. 


\bibliographystyle{ACM-Reference-Format}
\bibliography{reference}
\appendix
\section{Missing proofs in Section~\ref{app:large-delta}}
\subsection{Proof of Lemma~\ref{lem:app-distribution-domination}}
\label{app:proof-of-domination}
\begin{proof}
As shown in Lemma~\ref{lem:app-monotone}, it suffices to prove only for quantile interval $[p_u,1-p_u]$.
We let $c=\frac{L\Delta}{N}$, denote $q^{D_i}(v)$ as $q^D$ and $q^{E_i}(v)$ as $q^E$. First we prove that $\bm{D} \succeq \tilde{\bm{E}}$.
\begin{itemize}
    \item [(i)]$[p_u,\Delta]$, we have $f(q)=2\sqrt{cq}$. By Lemma~\ref{lem:app-bernstein}, we get 
\begin{align*}
q^E&\le q^D+ \sqrt{\frac{2(q^D-a_j)(b_j-q^D)L}{N}} + \frac{L(b_j-a_j)}{N}\\
&\le q^D + \sqrt{\frac{2c\cdot q^D(\Delta - q^D)}{\Delta}} + c\\
&\le q^D + \sqrt{2c\cdot q^D(1 - q^D)} + c.
\end{align*}
Since the shading functions are monotone, in order to prove that $\bm{D}\succeq \tilde{\bm{E}}$, we only need to prove the case when the above holds with equality. So we get the following:
\begin{align*}
(1+2c)q^D &= q^E-\sqrt{4{q^E}^2-4(1+2c)(q^E-c)^2}\\
&=q^E-\sqrt{2cq^E(1-q^E)-c^2(1+2c)+4q^Ec^2}\\
&\ge q^E -\sqrt{4cq^E+4c^2}\\
&\ge q^E -\sqrt{4cq^E}-2c.
\end{align*}
Therefore, we get $q^D\ge q^E -2\sqrt{cq^E}-4c$, since $q^D<1$.\\
\item[(ii)] $(\Delta, \frac12]$, by Lemma~\ref{lem:app-bernstein} we have 
\begin{align*}
    q^E&\le q^D+ \sqrt{\frac{2(q^D-a_j)(b_j-q^D)\cdot L}{N}} + \frac{L(b_j-a_j)}{N}\\
    &\le q^D +\sqrt{\frac{2\cdot \dfrac12 q^D\cdot q^D\cdot L}{N}}+\frac{L}{N}q^D\\
    &= q^D +q^D\sqrt{\frac{L}{N}}+\frac LN\cdot q^D
\end{align*}
Also, we only need to prove the case when $q^E=q^D(1+\sqrt{\frac LN}+\frac{L}{N}).$
Noting that
$$q^{\tilde{E}}\le q^D\left(1-2\sqrt{\frac{L}{N}}\right)\left(1+\sqrt{\frac L N} +\frac LN\right)<q^D,$$ we get $q^D\ge q^{\tilde{E}}$.
\item[(iii)] $(1/2, 1-\Delta)$, by Lemma~\ref{lem:app-bernstein} we have 
\begin{align*}
    q^E&\le q^D+ \sqrt{\frac{2(q^D-a_j)(b_j-q^D)\cdot L}{N}} + \frac{L(b_j-a_j)}{N}\\
    & \le q^D+ \sqrt{\frac{2(1-q^D )(\frac{1-q^D}{2})\cdot L}{N}} + \frac{L(1-q^D)}{N}\\
    &= q^{D}+\left(\sqrt{\frac L N}+\frac{L}{N}\right)(1-q^D). 
\end{align*}
When equality holds, we need to prove the case when
$$
    q^E=q^D+(1-q^D)\left(\sqrt{\frac{L}{N}}+\frac LN\right)~.
$$
We have the following step by step, by mildly assuming that $\sqrt{\frac{L}{N}}\le \frac{1}{4}$:
\begin{align*}
    \sqrt{\frac L N }&\ge 2\left(\sqrt{\frac LN}\right)^3+ 3\frac{L}{N}\\
    \left[1-\left(\sqrt{\frac LN}+\frac LN\right)\right]\cdot 2\sqrt{\frac LN}&\ge \left(\sqrt{\frac LN}+\frac LN\right)  \\
    [(1-q^D)-(\sqrt{\frac LN}+\frac LN)(1-q^D)]\cdot 2\sqrt{\frac LN}&\ge (1-q^D)\left(\sqrt{\frac LN}+\frac LN\right) \\
    f(q^E)&\ge q^E- q^D.
\end{align*}
Thus we get $q^{\tilde{E}} \le q^E-f(q^E)\le q^D.$
\item[(iv)]
$(1-\Delta,1-p_u]$, by Lemma~\ref{lem:app-bernstein} we have 
\begin{align*}
    q^E&\le q^D+ \sqrt{\frac{2(q^D-a_j)(b_j-q^D)\cdot L}{N}} + \frac{L(b_j-a_j)}{N}\\
    &\le q^D+ \sqrt{\frac{2(q^D-1+\Delta)(1-q^D)\cdot L}{N}} + \frac{L\Delta}{N}\\
    &\le q^D+ \sqrt{2 q^D(1-q^D)\cdot c} + \frac{L\Delta}{N}
\end{align*}
The last inequality is because $q^D\le 1$ and thus $q^D-1+\Delta \le q^D \cdot \Delta$.
Then similar as the (i) part, when the equality holds, we have 
\begin{align*}
    (1+2c)q^D &= q^E-\sqrt{2cq^E(1-q^E)-c^2(1+2c)+4q^Ec^2}\\
    &\ge q^E-\sqrt{2c(1-q^E)-c^2(1+2c)+4c^2}\\
    &\ge q^E-\sqrt{4c(1-q^E)+4c^2}\\
    &\ge q^E-2\sqrt{c(1-q^E)}-2c.
\end{align*}
Since $q^D<1$, we have $q^{\tilde{E}}=q^E-4c-2\sqrt{c(1-q^E)}\le q^D$.
\end{itemize}
For $\tilde{\bm{E}} \succeq \tilde{\bm{D}}$, we also split the situation into four cases.
\begin{itemize}
    \item [(i)] $[0, \Delta]$, by Lemma~\ref{lem:app-bernstein} we have 
    \begin{align*}
    \vert q^E -q^D\vert\le  \sqrt{2cq^D(1-q^D)}+c\\
    q^E\ge q^D-\sqrt{2cq^D(1-q^D)}-c.
    \end{align*}
    Since the shading functions are monotone, we only need to prove the case when the above holds with equality. 
    Thus $q^E<q^D$ and we have
    \begin{align*}
        q^{\tilde{E}}&=q^D-f(q^E)-4c-\sqrt{2cq^D(1-q^D)}-c \\
        &\ge q^D-2f(q^D)-5c =q^{\tilde{D}}.\\
    \end{align*}
    Thus we get $q^{\tilde{E}}\ge q^{\tilde{D}}$.
    \item[(ii)] $(\Delta, \frac{1}{2}]$, by Lemma~\ref{lem:app-bernstein} we have $q^E\ge q^D(1-\sqrt{\frac{L}{N}}-\frac{L}{N})$. When equality holds, we have 
    \begin{align*}
     q^{\tilde{D}}&= \left(1-4\sqrt{\frac LN}\right)q^D -5c\\
     &\le [1-3\sqrt{\frac{L}{N}}+\left(\sqrt{\frac{L}{N}}\right)^3]q^D-5c\\
     &\le \left(1-\sqrt{\frac LN}-\frac LN\right)\left(1-2\sqrt{\frac{L}{N}}\right)q^D-5c\\
     &\le q^E\left(1-2\sqrt{\frac LN}\right)-5c \le q^{\tilde{E}}
    \end{align*}
    Thus we have $q^{\tilde{E}}\ge q^{\tilde{D}}$.\\
    \item[(iii)]
    $(1/2, 1-\Delta]$, by Lemma~\ref{lem:app-bernstein} we have 
    $$
    q^E\ge q^D-\left(\sqrt{\frac LN}+\frac{L}{N}\right)(1-q^D).
    $$
    When the equality holds, we have 
    \begin{align*}
        q^{\tilde{D}}&=q^D-4\sqrt{\frac{L}{N}}(1-q^D)-5c\\
        &=\left(1+4\sqrt{\frac{L}{N}}\right)q^D-4\sqrt{\frac{L}{N}}-5.
    \end{align*}
     \begin{align*}
        q^{\tilde{E}}&=\left[q^D-\left(\sqrt{\frac{L}{N}}+\frac{L}{N}\right)\cdot (1-q^D)\right]-2\sqrt{\frac{L}{N}}\left(1-q^D+\left(\sqrt{\frac{L}{N}+\frac{L}{N}}\right)(1-q^D)\right)-4c \\
        &= q^D-(k+k^2)(1-q^D)-2k(1+k+k^2)(1-q^D)-4c \quad\quad(\textrm{By letting $k=\sqrt{\frac LN}$})\\
        &=(1+3k+3k^2+2k^3)q^D-(3k+3k^2+2k^3)-4c.
    \end{align*}
    so we need to prove that 
    \begin{align*}
        &q^{\tilde{D}}\ge q^{\tilde{E}}\\
        i.e.\quad &(1+4k)q^D-4k -5c\le (1+3k+3k^2+2k^3)q^D-(3k+3k^2+2k^3) -4c\\
        i.e.\quad &3k^2+2k^3-k \le q^D(3k^2+2k^3-k)+c
    \end{align*}
    We can add a mild assumption that $k<\frac{-3+\sqrt{17}}{4}$ and thus $3k^2+2k^3-k < 0$ and the above inequality holds, and thus we have $q^{\tilde{D}}\le q^{\tilde{E}}$.  \\
    \item[(iv)]
    $(1-\Delta, 1-p_u]$, by the (iv) part in the proof of $\bm{D}\succeq\tilde{\bm{E}}$, we have the inequality in this interval
    $$
    q^E \ge q^D -\sqrt{2c(1-q^D)q^D}-c.
    $$
    When the equality holds, we have 
    \begin{gather*}
        q^{\tilde{D}}=q^D-4\sqrt{c(1-q^D)}-5c,\\
        q^{\tilde{E}}=q^D-\sqrt{2c(1-q^D)q^D}-c-2\sqrt{(1-q^E)c}-4c.
    \end{gather*}
    So we need to prove that 
    $$
    4\sqrt{c(1-q^D)}\ge \sqrt{2c(1-q^D)q^D}+2\sqrt{(1-q^E)c}.$$
    Consider the term 
    \begin{align*}
    \frac{1-q^D}{1-q^E}&=\frac{1-q^D}{1-q^D+\sqrt{2c(1-q^D)\cdot q^D}+c}\\
    &=\frac{1}{1+\sqrt{2c\frac{q^D}{1-q^D}}+\frac{c}{1-q^D}}.
    \end{align*}
    Since the above is decreasing in $q^D$, we get the minimum of $\frac{1-q^D}{1-q^E}$ is taken when $q^D =1-p_u$.
    Recall that $p_u=\Theta(\frac{L\Delta}{N})$ as in Lemma~\ref{lem:app-monotone}, we can set $q^D:= Zc$, for some constant $Z$. 
    We then have $\frac{1-q^D}{1-q^E}=\dfrac{Z}{Z+\sqrt{2Z}+1}$, by selecting a properly large $Z$, we could let $\frac{1-q^D}{1-q^E}\ge \frac 23$. Thus we get the following step by step: 
    \begin{align*}
        6(1-q^D)&\ge 4-4q^E\\
        (4-\sqrt{2})\sqrt{c(1-q^D)}&\ge 2\sqrt{(1-q^E)\cdot c}\\
        4\sqrt{c(1-q^D)}&\ge \sqrt{2c(1-q^D)q^D}+2\sqrt{(1-q^E)\cdot c}.
    \end{align*}
    And we have finished the proof.
\end{itemize}
By combining the eight parts, we conclude that $\bm{D}\succeq\tilde{\bm{E}}\succeq \tilde{\bm{D}}$.
\end{proof}
\subsection{Proof of Theorem~\ref{thm:all-Delta}}
\label{app:proof-of-larger-Delta}
\begin{lemma}
Choose $N_1=\tilde{O}(N^2n\Delta)$.  
For $j\ge 2$ and $a_j\le \frac{1}{n}$, choose 
$$N_j=\tilde{O}\left(N^2n\Delta (\sqrt{j}-\sqrt{j-1})^2+\sqrt{\frac{N^2n\Delta}{(j-1)}}\right)~.$$ 
For $a_j\in (\frac{1}{n},\frac{1}{2}]$, choose 
$$N_j=\tilde{O}\left(\frac{N^2}{j(j-1)} + \frac{N}{(j-1)}\right)~.$$ 
With probability $1-\delta$, for any value $v$, and any buyer $i$ we have: 
\[
|q^{E_i}(v)-q^{D_i}(v)|\le f(q^{D_i}(v))
~.\]
\end{lemma}
\begin{proof}

For the part we use Algorithm~\ref{alg:shade-algorithm}, the proof is trivial.
If suffices to consider the case of $\Delta > f(q^{D_i}(v))$. 
By Lemma~\ref{lem:app-bernstein} we have in round $j$:
\[
|q^{E_i}(v)-q^{D_i}(v)|\le \sqrt{\frac{(q^{D_i}(v)-a_{j-1})(a_j-q^{D_i}(v))}{N_j}} + \frac{\Delta}{N_j}
~.\]
The log factors are included in our choice of $N_j$. 

For $j=1$ we choose a proper $N_1=\tilde{O}(N^2n\Delta)$, and we can verify $|q^{E_i}(v)-q^{D_i}(v)|\le f(q^{D_i}(v))$ in this interval $[0,\Delta]$.
For $j\ge 2$, the first term of $f$ dominates the second term $\frac{1}{N^2n}$, unless $\Delta \le \frac{1}{N^2n}$. 
If this happens, we will have $f(q^D)\ge \Delta$, and this implies we are using Algorithm~\ref{alg:shade-algorithm}.

For the sake of simplicity, we write $q^D=q^{D_i}(v)$, $q^E=q^{E_i}(v)$ and $r=\frac{1}{4N^2n}$. 
\begin{enumerate}
    \item If $f(q^D)\le \frac{1}{n}$, it suffices to prove the following for any $q^D\in[a_{j-1},a_j]$:
    \[
    \sqrt{\frac{(q^D-a_{j-1})(a_j-q^D)}{N_j}}\le \sqrt{rq^D}~, \frac{\Delta}{N_j}\le \sqrt{rq^D}
    ~.\]
    The first term is implied by $N_j\ge \frac{1}{r}\frac{(a_j-q^D)(q^D-a_{j-1})}{q^D}$, whose right hand side is maximized when $q^D=\sqrt{a_{j-1}a_j}=\sqrt{j(j-1)}\Delta$.
    Thus we can choose $N_j=\frac{\Delta}{r}(\sqrt{j}-\sqrt{j-1})^2=\tilde{O}(N^2n\Delta (\sqrt{j}-\sqrt{j-1}))^2$. \\
    The second term is implied by $N_j\ge \sqrt{\frac{\Delta}{r(j-1)}}$. 
    So our choice of $N_j$ satisfies. 
    \item If $f(q^{D_i}(v))>\frac{1}{n}$, we have to prove:
    \[
    \sqrt{\frac{(q^D-a_{j-1})(a_j-q^D)}{N_j}}\le \frac{q^D}{2N}~, \frac{\Delta}{N_j}\le \frac{q^D}{2N}
    ~.\] 
    The first term is implied by $N_j\ge 4N^2\frac{(a_j-q^D)(q^D-a_{j-1})}{(q^D)^2}$, whose right hand side is maximized when $q^D=\frac{2a_ja_{j-1}}{a_j+a_{j-1}}$. 
    Thus we can choose $N_j\ge \frac{N^2}{(j-1)j}$. \\
    The second term is implied by $N_j\ge \frac{2N}{(j-1)}$. 
    So our choice of $N_j$ satisfies. 
\end{enumerate}
\end{proof}

\begin{lemma}
$\bm{D}\succeq\tilde{\bm{E}}\succeq\tilde{\bm{D}}$. 
\end{lemma}
\begin{proof}
The proof is covered by Lemma~\ref{lem:distribution-domination} and Lemma~\ref{lem:app-distribution-domination}, except for the case when $f(q) \ge \Delta$.

When $f(q) \ge \Delta$, $\shades(q):=\max\{0,q-\Delta\}$. 
First consider $v\in supp (E_i)$. 
By Algorithm~\ref{alg:shade-algorithm-appendix}, we only have one targeted sample $v$ in the quantile interval, and we have $q^{D_i}(v)\in [q^{E_i}(v)-\Delta,q^{E_i}(v)]$. 
This immediately implies $q^{E_i}(v)-\Delta\le q^{D_i}(v)$. 
Then, as $q-2f(q)$ is increasing in $q$ by our definition of $f$, we have:
\[
q^{D_i}(v)-2f(q^{D_i}(v))\le q^{E_i}(v)-2f(q^{E_i}(v))\le q^{E_i}(v)-\Delta~,
\]
because $f(q^{E_i}(v))\ge \Delta$. 

For those $v\notin supp (E_i)$, there exists adjacent $v^+,v^-\in supp (E_i)$, such that $v^+\ge v>v^-$ and $q^{E_i}(v)=q^{E_i}(v^+)$ by definition of quantile. 
Further define the new notation $\bar q^{E_i}(v^-)$ which represents the smallest right end-point of sampled quantile intervals that returns $v^-$. 
We have $q^{D_i}(v)\le q^{D_i}(v^-)\le \bar q^{E_i}(v^-)$ and $\bar q^{E_i}(v^-)-f(\bar q^{E_i}(v^-))=q^{E_i}(v^+)$. 
Then by $f(\bar q^{E_i}(v^-))\ge \Delta$ we have:
\[
q^{E_i}(v)-\Delta 
= q^{E_i}(v^+)-\Delta 
= \bar q^{E_i}(v^-) - f(\bar q^{E_i}(v^-))- \Delta 
\ge \bar q^{E_i}(v^-)-2f(\bar q^{E_i}(v^-))\ge q^{D_i}(v)-2f(q^{D_i}(v))
~.\]
\end{proof}

\begin{lemma}
The total number of targeted samples is $\tilde{O}(\max \{nN^2\Delta,N\})$. 
\end{lemma}
\begin{proof}
For the part applying Algorithm~\ref{alg:shade-algorithm}, the number of targeted samples is $\tilde O(N)$ by Lemma~\ref{lem:pinpoint-2}. 
Now consider other parts. In the following we will assume an upper bound of $j$: $j\le J$.

First for $j\ge 2$ and $a_j\le \frac{1}{n}$, we observe that:
\[
(\sqrt{j}-\sqrt{j-1})^2 = \frac{1}{(\sqrt{j}+\sqrt{j-1})^2}\le \frac{1}{4(j-1)}
~.\]
Further we have:
\[
\sum_{j=2}^J\frac{1}{\sqrt{j-1}} \le 1+\int_1^J \frac{1}{\sqrt{x}}dx = O(\sqrt{J}). 
~.\]
Then we have\footnote{Note that $1+\frac{1}{2}+\frac{1}{3}+\cdots+\frac{1}{k}=O(\log k)$.}
\[
\sum_{1\le j\le \min \{J,\frac{1}{n\Delta}\}} N_j=\tilde{O}\left(N^2n\Delta + \sqrt{N^2n\Delta \min\{J,\frac{1}{n\Delta}\}}\right)
=
\tilde{O}\left(N^2n\Delta + N\right)
~.\]

Then consider $a_j\in (\frac{1}{n},\frac{1}{2}]$. 
Observe that
\[
\sum_{j=1/(n\Delta)+1}^J \frac{1}{j(j-1)} = \sum_{j=1/(n\Delta)+1}^J \frac{1}{j-1}-\frac{1}{j} \le n\Delta
~.\]
Similarly we have
\[
\sum_{1\le j\le \min \{J,\frac{1}{n\Delta}\}} N_j=\tilde{O}\left(N^2n\Delta  + N \right)
~.\]


To sum up, the number of targeted samples is therefore $\tilde{O}(nN^2\Delta+N)$, which is equivalently $\tilde{O}(\max \{nN^2\Delta,N\})$. 
\end{proof}

Note that we can directly use the sandwich lemma in Section~\ref{sec:multi} (Lemma~\ref{lem:dskl}) because our $f$ and $\shaded$ remain the same. 

Finally, by similar arguments as in Section~\ref{sec:multi-proof}, we can finish the proof of Theorem~\ref{thm:all-Delta}.

\section{Missing proofs in Section \ref{sec:single}}\label{prf:missing_single}

In this section we give proofs for upper bounds and lower bounds of complexity in single buyer case. The upper bounds rely on constructions of specific $\epsilon$-approximate algorithms, and the lower bounds rely on constructions of distribution of the buyer's possible types, which is necessary for applying Yao's Minimax Principle (Lemma~\ref{lem:Yao}).

\subsection{Proofs for upper bounds in single buyer case}\label{prf:upper_single}

We prove upper bounds by constructing $\epsilon$-approximate algorithms. 
Since these algorithms differs when the prior $D$ is from different distribution families, we split the proofs for Theorem \ref{thm:single_upper_all}  and Theorem \ref{thm:single_upper_sample} into three parts: 
\begin{enumerate}
\item When $D$ is regular or MHR, we need $O(\log \epsilon)=\tilde{O}(1)$ targeted queries to achieve $(1-\epsilon)\opt{D}$ expected revenue.\\
For targeted sample complexity, we need at most $\tilde{O}(\max \{1,\min\{\Delta^2\cdot \epsilon^{-4}, \Delta\cdot \epsilon^{-3}\}\})$ targeted samples to achieve  $(1-\ep)\opt{D}$ revenue, if $D$ is regular. 
We need at most $\tilde{O}(\max \{1,\Delta^2 \epsilon^{-2}\})$ targeted samples to achieve  $(1-\ep)\opt{D}$ revenue, if $D$ is MHR.

\begin{proof}
We only prove the part in the targeted sample model for any $\Delta \in (0,1)$.
The reduction from the targeted sample model to the targeted query model is: uniformly draw a quantile from each targeted sample interval in the algorithm. 
We first consider the case with a small enough $\Delta$, then we try to figure out the sample complexity with larger $\Delta$. The proof of the former relies on the following algorithm which outputs a reserve price. 
For regular distributions, $q_t$ is $\epsilon$; for MHR distributions, $q_t$ is $1/e$. 
At the end we will prove that this algorithm returns a reserve price that approximates the optimal revenue closely enough with high probability. 

\begin{minipage}{1\linewidth}
\begin{algorithm}[H]
    \caption{Targeted sample algorithm for regular/MHR case}
    \begin{itemize}
    \item[1] Input $\Delta,q_t,N$. Initialize the quantile interval to be $[a_0,b_0]$, with $a_0=q_t$, $b_0=1-q_t$. Let $i=0$.
    \item[2] Targeted sample $N$ times at each of the following 5 quantile points that has not been sampled yet: $a_i$, $\frac{3a_i+b_i}4$, $\frac{a_i+b_i}2$, $\frac{a_i+3b_i}4$, $b_i$. (Here to targeted sample once at quantile $q$ is to targeted sample once in the interval $[q-\Delta/2,q+\Delta/2]$.)
    \item[3] Estimate the revenue at these 5 quantile points using the sample values ($\tilde R=q\cdot\tilde v_m$), where $\tilde v_m$ is the median of all targeted samples from $[q-\Delta/2,q+\Delta/2]$. 
    Find the quantile point $z_i\in \{a_i, \frac{3a_i+b_i}4,\frac{a_i+b_i}2,\frac{a_i+3b_i}4, b_i\}$ with the largest estimated revenue.
    \item[4] Update the interval in the new round: let $a_{i+1}=\max\{a_i,z_i-\frac{b_i-a_i}4\}$, $b_{i+1}=\min\{b_i,z_i+\frac{b_i-a_i}4\}$, finally $i=i+1$.
    \item[5] If $b_i-a_i>\epsilon$, return to step (2).
    \item[6] Pick the quantile point in $[a_i,b_i]$ that is closest to $\frac12$, take a sample at this point and output this sample value as the reserve price.
    \end{itemize}
\end{algorithm}
\end{minipage}

We first give the intuition of this algorithm. In each round we evenly divide the current interval into 4 smaller intervals, which generates 5 edge points. In regular/MHR case the revenue-quantile curve is concave, so the peak must be in the smaller interval that is adjacent to the edge point with the highest revenue. 
In this algorithm we try to find this highest edge point in every iteration, and then pick the 2 smaller intervals which are adjacent to this point and combine them to be the new interval (or only one when the highest edge point is at the two ends). Hence in every iteration, we decrease the length of the current interval by at least one half. When the interval is small enough, we can find a point inside it whose revenue is close enough to the optimal. So the algorithm halts after $\Theta(\log\epsilon^{-1})=\tilde{\Theta }(1)$ iterations, each consists of $5N$ times of targeted sampling.

\begin{lemma}\label{lem1}
Let $\Delta\le q(1-q)\frac{\epsilon}{\log \epsilon^{-1}},\forall q\in [q_t,1-q_t]$, and $N=1$.
If $D$ is regular/MHR, and we targeted one sample $\tilde v$ from interval $[q-\Delta/2,q+\Delta/2]$, the error of the estimated revenue at quantile $q$ (estimated by $\tilde R(q)=q\cdot \tilde v$), satisfies that:
$$\dfrac{|\tilde R(q)-R(q)|}{R(q)}\le \dfrac{\epsilon}{\log \epsilon^{-1}}.$$
\end{lemma}
\begin{proof}
For convenience we let $q_l=q-\Delta/2$ and $q_h=q+\Delta/2$, and the corresponding valuations be $v_l$ and $v_h$, w.r.t. true distribution. So we have $v_h\le \tilde v\le v_l$ as $v(q)$ is always decreasing. First let us see the upper bound, trivially we have:
$$\tilde R(q)=q\cdot \tilde v\le v_l\cdot q=\frac{R(q_l)}{q_l}\cdot q.$$
Notice that $R(1)=0$, so $R$ is a concave function, $R(q_l)<R(q)\cdot\frac{1-q_l}{1-q}$. 
Hence we have
$$\tilde R(q)<R(q)\cdot\frac{1-q_l}{1-q}\cdot \frac q{q_l} \le R(q)\cdot\frac{2+q\epsilon/\log \epsilon^{-1}}{2-(1-q)\epsilon/\log \epsilon^{-1}}<(1+\epsilon/\log \epsilon^{-1})\cdot R(q).$$
The otherside can be proved similarly:
\begin{align*}
    \tilde R(q)\ge v_h\cdot q=\frac{R(q_h)}{q_h}\cdot q&>R(q)\cdot \frac{1-q_h}{1-q}\cdot\frac q{q_h}\\
    &=R(q)\cdot\frac{2-q\cdot \epsilon/\log \epsilon^{-1}}{2+(1-q)\cdot \epsilon/\log \epsilon^{-1}}\\&>(1-\epsilon/\log \epsilon^{-1})\cdot R(q).
\end{align*}
\end{proof}

We notice the fact that each time we pick the largest evaluated revenue, the highest revenue in our selected sub-interval would not be the global maximal revenue. To quantify the influence of this fact on the final result, we define \textit{shrinking coefficient} ($c_e$), which is the ratio of the maximum revenue in our newly chosen sub-interval and the maximum revenue in the original interval. (This variable is calculated in each one iteration separately.) In the ideal case, $c_e=1$ as we managed to choose the sub interval that contains the maximum revenue. Generally, as the relative error of our samples are limited to $\epsilon/\log \epsilon^{-1}$ by the conclusion of lemma \ref{lem1}, we can prove the following lemma:
\begin{lemma}\label{lem:shrink}
In each round of iteration, we have $c_e=1-\Theta(\epsilon/\log \epsilon^{-1})$.
\end{lemma}
\begin{figure}[h]
    \centering
    \includegraphics[scale=0.6]{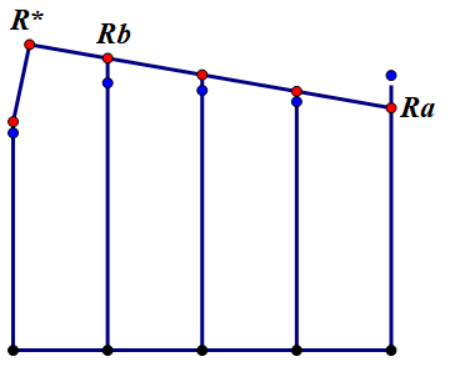}
    \setcaptionwidth{0.6\linewidth}
    \caption{The extreme case of the five evaluated revenue points and their true revenue points in each round. We denote the true revenue point as red points and our evaluated revenue point as the blue points. }
        \label{fig:1Hbounded}
\end{figure}

\begin{proof}
Suppose our selected sampled quantile point with the largest evaluated revenue has true revenue $R_a$ and the highest revenue point on the original interval is with revenue $R^*$. We suppose that its most adjacent sampled quantile point on $R_a$ point side is with revenue $R_b$. Since we have picked $R_a$ as the largest evaluated point instead of $R_b$ point, by Lemma~\ref{lem1} we have the following inequality:
$$
R_a(1+\epsilon/\log\epsilon^{-1}) \ge R_b(1-\epsilon/\log\epsilon^{-1}).
$$
Suppose that there are $d$ intervals between $R_a$ and $R_b$, by the concavity of the $R-q$ curve, we have 
$$
R^* \le R_b + (R_b -R_a)/ d.
$$
We denote the maximal revenue in the selected new interval as $R_n$, by the concavity of the R-q function, we get $R_n \ge R_a + (R_b-R_a)/d$ 
Then we could derive the lower bound of $c_e$:
$$
c_e =\frac{R_n}{R^*} \ge \frac{R_a+ (R_b-R_a)/d}{R_b+ (R_b- R_a)/ d}.
$$
Denote $R_b/R_a$ by $r$, then we have $$c_e\ge \frac{1+(r-1)/d}{r+(r-1)/d}.$$
By partial derivative with respect of $r$ and $d$, we find that the lower bound of $c_e$ is taken when $r= \dfrac{1+\epsilon/\log\epsilon^{-1}}{1-\epsilon/\log\epsilon^{-1}}$ and $d=3$ (just as the figure shows) and its minimum value is $$\frac{3-\epsilon/\log\epsilon^{-1}}{3+5\epsilon/\log\epsilon^{-1}}=\frac{1-\epsilon/3\log\epsilon^{-1}}{1+5\epsilon/3\log\epsilon^{-1}}\ge 1-2\epsilon/\log\epsilon^{-1}.$$

\end{proof}
 
 

So after $\Theta(\log\epsilon^{-1})$ rounds, the total shrinking coefficient is at least:
$$C_e=\prod c_e=1- O(\epsilon).$$

Up to now we have found a small quantile interval with length $\epsilon$, and the highest revenue inside it is at least $(1-C_e)$ times of the optimal revenue. In the last step of our algorithm, query for the sample value at the point closest to 0.5 inside this interval. (In most cases, we have already done a querying at this point, the result can be directly used again.) Then according to concavity of the revenue function, the revenue of this point is at least $1-\epsilon$ times the highest value in this interval. So the final revenue at our returned value still have a $1-O(\epsilon)$ approximation. 

For regular distributions, we have $q_t=\epsilon$ and thus $\Delta$ in Lemma \ref{lem1} is $\tilde O(\epsilon^2)$. 
For MHR distributions, we have $q_t=\frac{1}{e}$ and thus $\Delta$ in Lemma \ref{lem1} is $\tilde O(\epsilon)$. 
Therefore, for such a small $\Delta$, the targeted sample complexity is $\tilde O(1)$. 

Then we focus on the case when $\Delta$ is large, all our algorithm is almost the same as the above. However, if we still sample once in each interval with length $\Delta$, we could not get a good enough approximation for the revenue at this quantile with at most $\epsilon$ multiplicative revenue loss to satisfy the lemma \ref{lem1}. In order to upper bound the revenue loss with $\epsilon/\log(\epsilon^{-1})$, we need to sample $N$ times in each interval with length $\Delta$. 
We again use Lemma~\ref{lem:app-bernstein} to prove this. 
\begin{itemize}
    \item [1.] When $\Delta=\tilde\Omega(\epsilon^2)$ and $\Delta = \tilde O(\epsilon)$, we know that the gap between quantiles of the empirical distribution and the true distribution is $\tilde O(\Delta/\sqrt{N})$, for any specific value. We need this gap to be smaller than $q \cdot \epsilon= \tilde{\Omega}(\epsilon^2)$ (since $q\ge \ep$), thus we have $\tilde O(\Delta/\sqrt{N})\le \tilde \Omega({\epsilon^2})$ and we need $N= \tilde{O}(\Delta^2\cdot \epsilon^{-4})$ samples.
    \item[2.] When
    $\Delta=\tilde{\Omega}(\epsilon)$, if this sampling interval is the first interval from $0$ to $\Delta$, we could map this interval to the interval $[0,1]$ and it is still a regular distribution and we only need $O(\Delta\cdot \epsilon^{-3})$ to estimate a $(1-\epsilon)$ multiplicative largest revenue value in this interval.
    If the interval is not the first interval, then we have $q^D>\Delta$ and we only need the quantile error to be smaller than $\Delta\cdot \epsilon$. Then we should ensure $\tilde O(\Delta/\sqrt{N})\le \Delta\cdot \epsilon$, thus $N=\tilde O(\Delta\ep^{-3})$ is sufficient.
\end{itemize}
Thus in the regular case, we only need to query in each interval for $N$ times and also use $\tilde{O}(1)$ rounds of binary search queries to make the final interval be small enough, and we finally sample $N$ times in this interval to get an $1-\epsilon$ approximation of the optimal revenue since there are only $\tilde {O}(1)$ rounds of revenue loss and each of them could be smaller than $O(\epsilon)$ by arbitrary log and constant factors. Combining these together, we get the final sample complexity is $\tilde {O}(\max \{1,\min\{\Delta^2\ep^{-4},\Delta\ep^{-3}\}\})$.\\

For MHR cases with large $\Delta$, since at the first step of our algorithm we truncate the first and last $1/e$ amount of quantile, the related quantile is always $\Theta(1)$ in our algorithm. So in order to achieve an $\epsilon$-multiplicative error in each round, we should have $\tilde O(\Delta/\sqrt{N})\le O(\ep)$, so $N =O(\Delta^2 \epsilon^{-2})$ is enough for the approximation. 
\end{proof} 

\item When $D$ is {$[0,1]$}{-bounded}, we only need $O(\epsilon^{-1})$ targeted queries to achieve $\opt{D}-\epsilon$ expected revenue.\\

For targeted sample complexity, we need at most $\tilde{O}(\max \{\epsilon^{-1}, \Delta\cdot \epsilon^{-2}\})$ targeted samples to achieve  $\opt{D}-\ep$ revenue. Note that when $\Delta =O(\ep)$ the two complexities are the same.
\begin{proof}
The reduction from the targeted sample model to the targeted query model is: uniformly draw a quantile from each targeted sample interval in the algorithm. 
Our algorithm is to divide the whole quantile space into small intervals of length $d$. 
First consider the case when $\Delta$ is small enough. 
For each quantile $q(q=kd,k\in\mathbb{N},k\le 1/d)$, we targeted sample once from $[q-e, q+e](e\le d/2)$. Then we estimate the revenue of the quantile $q$ by the returned value, say, $\tilde v(q)$, times $q$. Among all the values returned, we select the value with the largest estimated revenue, namely $\tilde v_s$, as the reserved price. We suppose the corresponding quantile of $\tilde v_s$ w.r.t. true distribution is $\tilde q_s$ and $\tilde v_s$ is returned when we query around the quantile $q_s$. Suppose the optimal quantile and its corresponding value are $q^*$ and $v^*$. 
We first show that, the revenue of this algorithm $\tilde v_s \cdot \tilde q_s$ is close to optimal.
\begin{lemma}
$\tilde v_s \cdot \tilde q_s \ge q^*\cdot v^* - 2d -3e$. 
\end{lemma}
\begin{proof}
If $q^*<2d+3e$, the lemma holds trivially.
Otherwise we select the quantile which is the largest queried quantile among all the quantiles at least $d$ smaller than $q^*$ and we denote it as $q_b$ and we denote its corresponding quantile as $q_{b}$. When we targeted sample around the quantile $q_b$, suppose we get a returned value $\tilde v_b$ and the corresponding quantile of $\tilde v_b$.
We have the following inequalities:
\begin{gather*}
    q_b -e \le \tilde q_b \le q_b + e,\quad q_s-e\le \tilde q_s\le q_s+e\\
    \quad q^* -2d - e  <q_b < q^* -d +e, \\
    q_s \cdot \tilde v_s \ge q_b \cdot \tilde v_b,\\
    \tilde q_b < q^*,\quad \tilde v_b >v^*.\\
\end{gather*}
Therefore we have the inequality
\begin{align*}
    \tilde q_s\cdot \tilde v_s &\ge (q_s -e)\cdot \tilde v_s \\
    &\ge q_b\cdot \tilde v_b -e \cdot \tilde v_s\\
    &\ge (q^* -2d -e)\cdot v^* -e(\tilde v_b +\tilde v_s)\\
    &\ge q^*\cdot v^* -(2d+e)\cdot v^* - e\cdot \tilde v_b -e\cdot\tilde v_s\\
    &\ge q^*\cdot v^* -2d-3e.
\end{align*}
\end{proof}
Thus if we let $d=\epsilon/4$ and $e=\epsilon/8$, $\Theta(\epsilon^{-1})$ queries are enough to get a $1-\epsilon$ approximation of maximum auction revenue, under the condition that $\delta \le \epsilon/4$. 
This covers the case when $\Delta=O(\epsilon)$. 
For larger $\Delta$, we combine this result together with the result in the multi-buyer setting with $n=1$, we can know the targeted sample complexity for the $[0,1]$-bounded distribution is $\tilde{O}(\max\{\epsilon^{-1}, \Delta\epsilon^{-2}\})$. 
\end{proof}
\item When $D$ is {$[1,H]$}{-bounded}, we only need $O(\log H\cdot \epsilon^{-1})$ targeted queries to achieve $(1-\epsilon)\cdot \opt{D}$ expected revenue.\\
The targeted sample complexity is $\tilde{O}(\max\{\Delta H\ep^{-2},H^{1/2}\epsilon^{-1}\})$. 
\begin{proof}
We also prove in the targeted sample model, and firstly consider a small enough $\Delta$. The reduction is same as the previous part. Our algorithm is to sample according to quantile series $\{q_1= 1/H,q_2=(1+\epsilon/8)/H, q^3=(1+\epsilon/8)^2/H,\cdots, q_m=1\}$, each quantile $q_i$ corresponds to the sampling interval $[q_i-\epsilon/8H,q_i+\epsilon/8H]$ and we only take one sample $v_i$ from this quantile interval. The targeted sample complexity would be $O\left(\frac{\log H}{\log (1+\epsilon/8)}\right)=\tilde{O}( \epsilon^{-1})$.

For each quantile $q_i$ with returned sample value $v_i$, we calculate $q_i\cdot v_i$ and then we select the $v_i$ with the largest $q_i\cdot v_i$ as the reserved price. We let $q(v)$ be the quantile function with respect to value $v$.

We first prove that the selected value is at least $(1-\epsilon/4)$ revenue than any of the revenues of the returned values. Suppose that the algorithm finally select the returned value $\tilde{v_s}$ and its corresponding queried $q_i$ is $q_s$ (this is known to the algorithm), its true quantile is $\tilde{q_s}=q(\tilde{v_s})$. Then for any returned $v_i$, since we pick the largest $v\cdot q$ as $v_s$, we have 
$$
    \tilde{v_s}\cdot q_s\ge {v_i}\cdot q_i.
$$
Since $\vert q_i-q(v_i)\vert\le \epsilon/8H$ and $\vert q_s-\tilde{q_s}\vert \le \epsilon/8H$, we get 
\begin{align*}
    \tilde{v_s}\cdot \tilde{q_s} &\ge \tilde{v_s}\cdot (q_s-\epsilon/8H)\\
    &=\tilde{v_s}\cdot q_s\cdot \frac{q_s-\epsilon/8H}{q_s}\\
    &\ge{v_i}\cdot q_i \cdot \frac{q_s-\epsilon/8H}{q_s}\\
    &\ge {v_i}\cdot q(v_i)\cdot \frac{q_i}{q_i+\epsilon/8H}\cdot \frac{q_s-\epsilon/8H}{q_s}\\
    &\ge {v_i}\cdot q(v_i) \cdot 
    \frac{1/H}{1/H+\epsilon/8H}\cdot \frac{1/H-\epsilon/8H}{1/H}\\
    &\ge {\left(1-\dfrac{\epsilon}4 \right)} \cdot {v_i}\cdot q(v_i).
\end{align*}
Suppose the global optimal value is $v^*$ and its corresponding quantile is $q(v^*)=q^*$, we let the maximum $q_i$ less than $q^*$ be $q_b$, we have $q_b\cdot (1+\epsilon/8)\ge q^*\ge q_b$ then we get 
$q^* \le q_b(1+\epsilon/8)\le q(v_b)(1+\epsilon/4)(1+\epsilon/8)$ and $q(v_b)\cdot v_b\ge (1-\epsilon/2) \cdot q^*\cdot v^*$. Finally we get 
\begin{align*}
\tilde{q_s}\cdot \tilde{v_s} &\ge (1-\epsilon/4)\cdot v_b\cdot q(v_b)\\
&\ge (1-\epsilon/4)\cdot (1-\epsilon/2)\cdot q^*\cdot v^*    \\
&\ge (1-\epsilon)\cdot q^*\cdot v^*
\end{align*}
And we have proved that $O(\log H\cdot \epsilon^{-1})$ samples with $\Delta =O(\epsilon/H)$ is enough to get an $(1-\epsilon)$-multiplicative approximate revenue in $[1,H]$-bounded distributions.
\end{proof}

For larger $\Delta$, we combine these results with our multi-buyer results when $n=1$, we get the upper bound of targeted sample complexity with $[1,H]$ distributions is $\tilde{O}(\max\{\Delta H\ep^{-2},H^{1/2}\epsilon^{-1}\})$.
\end{enumerate}

\subsection{Proofs for lower bounds in single buyer case}\label{prf:lower_single}

To prove a lower bound with help of Yao's Minimax Principle(Lemma~\ref{lem:Yao}), we only need to show that there exists a fixed distribution of the bidder's value distributions, such that no deterministic algorithm can return a reserve price with less expected revenue loss than claimed in the theorem. Based on this idea, the followed proof of theorem \ref{thm:single_lower_all} is split into three parts under different conditions, each with corresponding construction of such distribution. For generality of the proof, we always assume strongest targeting power of the querier: he/she can get value at  exact quantile points, which is to say we allow $\Delta = 0$.

\begin{enumerate}
\item When $D$ is \textbf{regular}, we need at least
$\Theta(\log(\epsilon^{-1}))$ targeted queries to achieve at least $(1-\epsilon)\opt{D}$ expected revenue.
\begin{figure}
        \centering
        \begin{minipage}[t]{0.48\textwidth}
        \centering
        \setcaptionwidth{0.8\textwidth}
       
       \includegraphics[scale=0.4]{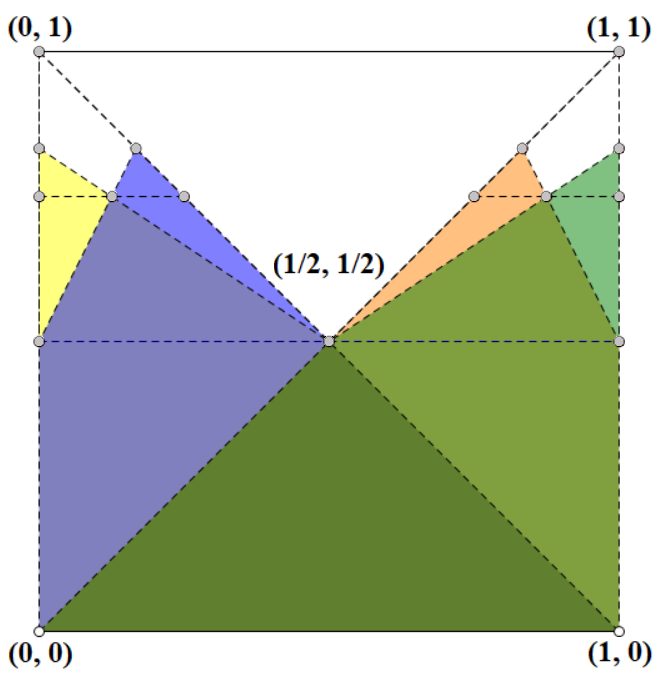}
        \caption{$R-q$ space for the case when $s=2$. The whole square is with four corners $(0,0),(1,1),(0,1),(1,0)$. The randomized distribution consists of four different-coloured distributions each w.p. 1/4.}
        \label{fig:s=2ex}
        \end{minipage}
        \centering
        \begin{minipage}[t]{0.48\textwidth}
        \setcaptionwidth{0.9\textwidth}
        \includegraphics[scale=0.45]{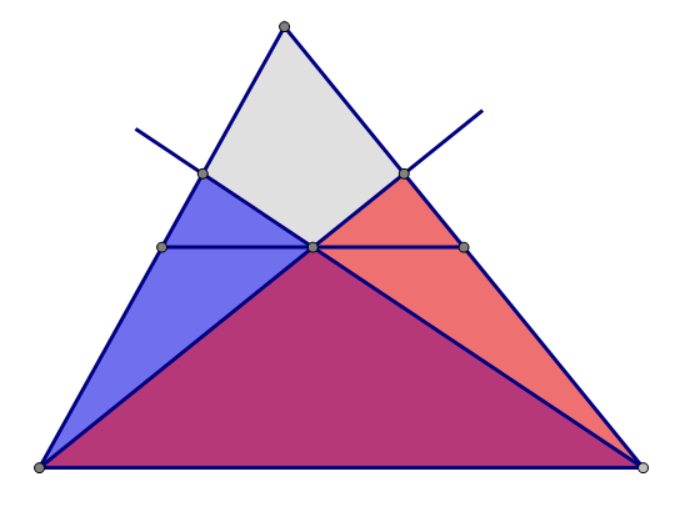}
        \caption{The whole triangle is the original top-triangle of case $2^{s-1}$, and the blue and orange distributions are the newly built distributions. The two new top-triangles of case $2^s$ are the two smaller triangles above the median. }
        \label{fig:top-tri}
        \end{minipage}
\end{figure}
\begin{proof}

For applying Lemma~\ref{lem:Yao}, here gives the construction of the distribution of the bidder's possible types: it contains $2^s$ possible $R-q$ curves where $s=\log_6{\frac{1}{2\sqrt{\epsilon}}}$, each occurs with same probability. Figure \ref{fig:s=2ex} shows an example for small $s$. 

    Further, the construction of the randomized distribution is done iteratively. 
    The distribution with $2^s$ possibilities is based on the $2^{s-1}$ case: every curve of the latter case is transformed into 2 new curves, by doing a specific transformation on the `top triangle': First we draw the median parallel to the $x$-axis and take its midpoint. Then we connect two vertices on the bottom margin to this midpoint and extend them until they intersect with the other two margins of the original triangle. The new-generated `top triangles' are the two smaller non-overlapped triangles on the top with their bottom margin parralel to the $x$-axis. See Fig.~\ref{fig:top-tri} for an illustration.
    
    In our construction, the regular condition is always satisfied. Now suppose we are allowed to query for $k$ times, and we try to lower bound the expected revenue loss under this condition. There are $2^s$ values that are possible to be the optimal. And from the distribution construction it is easy to see that if the choice is wrong, the revenue loss is at least $2\sqrt{\epsilon}$. 
    
    Now suppose the algorithm adaptively takes $k$ samples, then with probability at least $2^{-k}$ it still cannot find the optimal point of the curve. (The algorithm unluckily guesses wrong in every round of the binary search.) Thus it has to guess among all the possible top points left, with succeeding probability of at most $2^{k-s}$. Thus the failing probability of querying $k$ samples is at least $p_l=2^{-k}\cdot (1-2^{k-s})=2^{-k}-2^{-s}.$
    So the lower bound of expected revenue loss is $p_l\cdot 2\sqrt{\epsilon}$. To restrict the multiplicative loss, we should have the following:
    \begin{align*}
        p_l \cdot 2\sqrt{\epsilon} &\le \epsilon\\
        2^{-k}-2^{-s}&\le \frac{\sqrt{\epsilon}}{2}\\
        2^{-k} &\le \frac{\sqrt{\epsilon}}{2}+2^{-\log_6{\frac{1}{2\sqrt{\epsilon}}}}\\
        &\le \frac{1}{2}\cdot \epsilon^{\frac{1}{2\log_2^6}}+2^{-\log_6{\frac{1}{2\sqrt{\epsilon}}}}\\
        &\le (1/2+ 2^{\frac{1}{\log_2^6}})\cdot \epsilon^{\frac{1}{2\log_2^6}}
    \end{align*}
    Thus we get the lower bound of queries
    \begin{align*}
        k&\ge \frac{1}{2\log_2 6}\log (\epsilon^{-1})-\log_2 (1/2+ 2^{\frac{1}{\log_2^6}})\\
        &=\Theta(\log \epsilon^{-1}).
    \end{align*}
    We have finished the proof.
\end{proof}

\item
        When $\bm{D}$ is \bm{$[0,1]$}\textbf{-bounded}, we need at least $\Theta(\epsilon^{-1})$ targeted queries to achieve at least $\opt{D}-\epsilon$ expected revenue.
        
\begin{proof}

    

For applying Lemma~\ref{lem:Yao}, here gives the construction of the distribution of the bidder's possible types:
Suppose that each of the support distribution of the bidder's value is denoted as $D_s$ where $s$ is an integer in $[0,\frac{1}{2\epsilon})$. The revenue-quantile curve under $D_s$ is defined as:
$$
R^{D_s}(q)=\left\{
\begin{aligned}
&q,\quad &(0\le q< \frac12)\\
&\frac1{2}\cdot (q+\frac12-4s\epsilon),\quad &(\frac12+4s\epsilon\le q<\frac12+4(s+1)\epsilon)\\
&\frac12,\quad &(\text{else where})\\
\end{aligned}\right.
$$
    
Fig.~\ref{fig:01boudned} shows the $R-q$ curve of $D_s$. 
    \begin{figure}[h]
        \centering
        \includegraphics[scale=0.4]{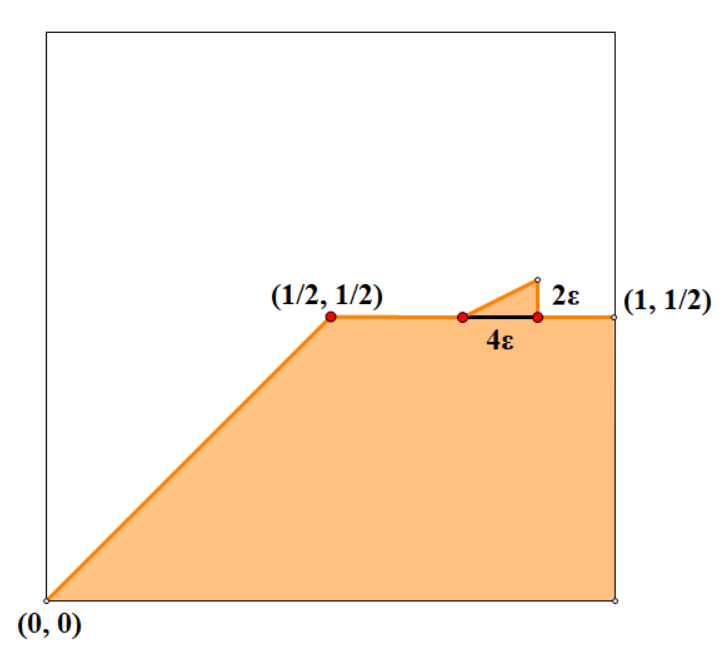}
        \setcaptionwidth{0.8\linewidth}
        \caption{The $R-q$ curve of one $D_s$ we constructed for some $s$. The hill is in interval $[\frac12+s\epsilon,\frac12+(s+1)\epsilon) $. Note that the slope of the small hill is always 1/2, so the $R-q$ curve is well defined.}
        \label{fig:01boudned}
    \end{figure}
    We let $s$ follows a uniform distribution in the set of integers in $[0,\frac{1}{2\ep})$. 
    This forms a randomized distribution of the bidder's value based on our construction of $D_s$. 
    Now we analyse how many samples are necessary for a deterministic algorithm to approximate the optimal result with an additive error of at most $\ep$.
    
    First let us see how this distribution's $R-q$ curve looks like. In the quantile interval $[\frac12,1]$ the $R-q$ curve is almost horizontal at an altitude $\frac12$, except some small `hills' in small sub-intervals, each with a height of at least $\frac{1+\epsilon}2$ (as illustrated in Fig.~\ref{fig:01boudned}). 
    The hill might occur at $\frac1\epsilon$ different places. If an algorithm fail to report a value at the hill, then the approximation fails. 
    Now suppose the algorithm takes $k$ samples, then it can see whether there exists a hill at at most $k$ places. 
    Obviously the optimal algorithm based on these $k$ samples must have the following form: (1) query at $k$ places where the hill might occur. (2) return the corresponding value if it sees the hill, or return another value at other places that has not been queried. 
    The probability of failing is at least:
    $$p_f=\frac{\frac1{8\epsilon}-k}{\frac1{8\epsilon}}\cdot\left(1-\frac1{\frac1{8\epsilon}-k}\right)=1-8\epsilon(k+1).$$
    When the algorithm fails, then there would be at least an $2\epsilon$ revenue loss, so the expected revenue loss is at least $(1-8\epsilon(k+1))\cdot 2\epsilon$, we need this revenue loss less than $\epsilon$, then we have $1-8\epsilon(k+1)\le 1/2$ and $k>\frac1{16}\cdot \epsilon^{-1} -1 =\Theta(\epsilon^{-1})$
\end{proof}
\item
When $\bm{D}$ is \bm{$[1,H]$}\textbf{-bounded}, we need at least $\Theta(\epsilon^{-1}\cdot H)$ targeted queries to achieve at least $(1-\epsilon)\opt{D}$ expected revenue.\\

The targeted sample complexity when $\Delta=\tilde O(\epsilon/H).$
\begin{proof}
For applying Lemma~\ref{lem:Yao}, here gives the construction of the distribution of the bidder's possible types.\\
\begin{figure}[h]
        \centering
        \includegraphics[scale=0.4]{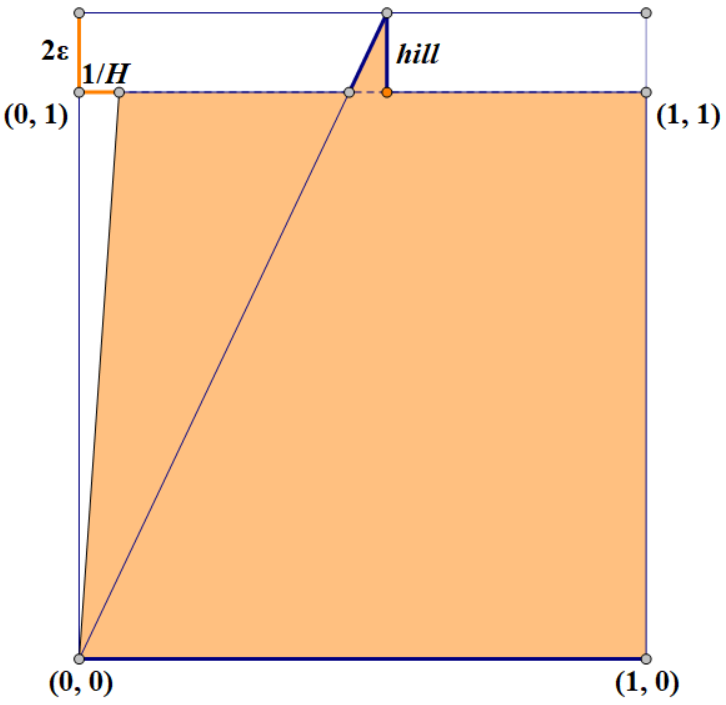}
        \setcaptionwidth{0.8\linewidth}
        \caption{The $R-q$ curve of one $D_s$ we constructed for some $s$. The hill is in interval $\left[\frac{(1+2\epsilon)^{s-1}}{H},\frac{(1+2\epsilon)^{s}}{H}\right)$. Note that the slope of the small hill is always 1/2, so the $R-q$ curve is well defined.}
        \label{fig:1Hbounded}
    \end{figure}
Similar as the previous part, we first give a family of distributions and then the randomized distribution is uniformly picked from this family. Any distribution in this family is denoted as $D_s$, then the $R-q$ curve under $D_s$ is defined as: 
$$
R^{D_s}(q)=\left \{
\begin{aligned}
&Hq \quad& (0\le q< 1/H)\\
&H(1+2\epsilon)^{-i+1}q \quad& \left(\frac{(1+2\epsilon)^{s-1}}{H}\le q< \frac{(1+2\epsilon)^s}{H}\right)\\
&1\quad&\text{(else where)}
\end{aligned}\right..
$$
One of the constructed distribution is shown in figure \ref{fig:1Hbounded}.

Our randomized distribution is built as a uniform mixture of all the above distributions in this family. So each of he support distribution has maximal revenue $1+2\epsilon$ and there are $\frac{\log H}{\log (1+2\epsilon)}$ distributions. Similar as the $[0,1]$-bounded case each distribution has a small hill whose revenue is larger than 1 (the second part of the revenue) and the hill is $2\epsilon$ high. So if an algorithm queried at most $k$ times, and it sets a reserved price, it has probability $(K-k-1)/K$ to select the wrong hill as the reserved price. In order to give an expected revenue loss less than $\epsilon$, we need the algorithm query at least $K/2-1=O(\frac{\log H}{\log (1+2\epsilon)})= \Theta(\log H\cdot \epsilon^{-1})$ times.
\end{proof}
\end{enumerate}

\end{document}